%% file: main.tex
\PassOptionsToPackage{svgnames}{xcolor}
\documentclass[acmsmall]{acmart}
\acmJournal{PACMPL}
\acmVolume{1}
\acmNumber{CONF} 
\acmArticle{1}
\acmYear{2018}
\acmMonth{1}
\acmDOI{} 
\startPage{1}

\setcopyright{none}

\input{settings.tex}
\toggletrue{fullversion}
\input{macros.tex}

\usepackage{graphicx}
\usepackage[UKenglish]{babel}
\begin{document}
\newtheorem{remark}[theorem]{Remark}
\title{Statically Verified Refinements for Multiparty Protocols}

\author{Fangyi Zhou}
\orcid{0000-0002-8973-0821}             
\affiliation{
  \institution{Imperial College London}            
  \country{United Kingdom}                    
}
\email{fangyi.zhou15@imperial.ac.uk}          

\author{Francisco Ferreira}
\orcid{0000-0001-8494-7696}             
\affiliation{
  \institution{Imperial College London}            
  \country{United Kingdom}                    
}
\email{f.ferreira-ruiz@imperial.ac.uk}          

\author{Raymond Hu}
\orcid{0000-0003-4361-6772}             
\affiliation{
  \institution{University of Hertfordshire}            
  \country{United Kingdom}                    
}
\email{r.z.h.hu@herts.ac.uk}          

\author{Rumyana Neykova}
\orcid{0000-0002-2755-7728}             
\affiliation{
  \institution{Brunel University London}            
  \country{United Kingdom}                    
}
\email{rumyana.neykova@brunel.ac.uk}          

\author{Nobuko Yoshida}
\orcid{0000-0002-3925-8557}             
\affiliation{
  \institution{Imperial College London}            
  \country{United Kingdom}                    
}
\email{n.yoshida@imperial.ac.uk}          

\input{abstract.tex}

\maketitle              
\input{intro.tex}
\input{overview.tex}

\input{implementation.tex}
\input{theory.tex}

\input{evaluation.tex}
\input{related_work.tex}
\input{conclusion.tex}
\input{acks.tex}
\bibliography{fangyi}
\iftoggle{fullversion}{
\begin{appendix}
\section{Proofs for \S~\ref{section:theory}}
\input{proof.tex}
\input{implementation_appendix.tex}
\end{appendix}
}{}

\end{document}

%% file: settings.tex
\usepackage{xcolor} 
\definecolor{lightgray}{gray}{0.85}
\usepackage{pifont}

\usepackage[utf8]{inputenc}
\usepackage[T1]{fontenc}
\usepackage{mathtools}
\usepackage{amsmath}
\usepackage{graphicx}
\usepackage{subcaption}
\usepackage{xcolor}
\usepackage{listings}
\usepackage{wrapfig}
\usepackage{etoolbox}
\usepackage{caption}
\usepackage{tikz}
\usetikzlibrary{snakes,arrows,shapes,shapes.multipart,decorations.pathreplacing,calc,positioning}
\usepackage{float}
\usepackage{xspace}
\usepackage[capitalise]{cleveref}
\usepackage{tabu}
\usepackage[inline]{enumitem}
\usepackage{multirow}
\usepackage{mathpartir}
\usepackage{dashbox}

\newtoggle{fullversion}

\crefname{section}{\S}{\S\S}

\definecolor{keywordcolour}{rgb}{0.5,0,0.35}
\definecolor{greencomments}{rgb}{0,0.5,0}
\lstdefinelanguage{scribble}{
  basicstyle=\small\ttfamily,
  stringstyle=\color{Blue},
  showstringspaces=false,
  keywords={and,as,at,by,catches,choice,continue,do,from,global,import,instantiates,interruptible,local,module,or,par,protocol,rec,role,sig,throws,to,type,with,int,aux},
  morestring=[b]",
  morestring=[b]',
  morecomment=[l][\color{greencomments}]{//},
  literate={->}{{${\rightarrow\ }$}}1 {>}{{$>$}}1 {<}{{$<$}}1 {<=}{{$\leq$}}1 {>=}{{$\geq$}}1 {&&}{{$\land$}}1 {||}{{$\lor$}}1 {!=}{{$\neq$}}1 {=}{{$=$}}1
}
\lstdefinelanguage{fsharp}%
{morekeywords={let, new, match, with, rec, open, module, namespace, type, of, member, %
and, for, while, true, false, in, do, begin, end, fun, function, return, yield, try, %
mutable, if, then, else, cloud, async, static, use, abstract, interface,
inherit, finally, val, int, unit, noeq, ML },
  basicstyle=\small\ttfamily,
  otherkeywords={ let!, return!, do!, yield!, use!, var, select, where,
    order, by, bool, option, EntryPoint, Refined },
  keywordstyle=\color{keywordcolour},
  sensitive=true,
  breaklines=true,
  tabsize=4,
  morecomment=[l][\color{greencomments}]{///},
  morecomment=[l][\color{greencomments}]{//},
  morecomment=[s][\color{greencomments}]{{(*}{*)}},
  morestring=[b]",
  showstringspaces=false,
  literate={->}{{${\rightarrow\ }$}}1 {>}{{$>$}}1 {<}{{$<$}}1 {<=}{{$\leq$}}1 {>=}{{$\geq$}}1 {&&}{{$\land$}}1 {||}{{$\lor$}}1 {!=}{{$\neq$}}1 {<>}{{$\neq$}}1 {:=}{{$\coloneqq$}}1 {=}{{$=$}}1
}

\lstdefinelanguage{raw}{
  morekeywords={},
  otherkeywords={},
}

%% file: macros.tex
\newcommand{\Ourtool}{\texorpdfstring{\textsc{Session$^{\star}$}}{SessionStar}\xspace}

\newcommand{\mypara}[1]{\paragraph{\textbf{#1}}}

\newcommand{\code}[1]{\texttt{\upshape #1}}
\newcommand{\rulename}[1]{\textsc{#1}}
\newcommand{\myinferrule}[3][]{\ensuremath{\inferrule*[left={#1}, sep=1em]{#2}{#3}}}

\newcommand{\ruleTE}[1]{[\rulename{TE-#1}]}
\newcommand{\ruleTEVar}{\ruleTE{Var}}
\newcommand{\ruleTEPlus}{\ruleTE{Plus}}
\newcommand{\ruleTESub}{\ruleTE{Sub}}
\newcommand{\ruleTEConst}{\ruleTE{Const}}

\newcommand{\ruleE}[1]{[\rulename{E-#1}]}
\newcommand{\ruleEPhi}{\ruleE{Phi}}
\newcommand{\ruleECtx}{\ruleE{Cnt}}
\newcommand{\ruleERec}{\ruleE{Rec}}

\newcommand{\ruleG}[1]{[\rulename{G-#1}]}
\newcommand{\ruleGPfx}{\ruleG{Pfx}}
\newcommand{\ruleGCtx}{\ruleG{Cnt}}
\newcommand{\ruleGRec}{\ruleG{Rec}}

\newcommand{\ruleL}[1]{[\rulename{L-#1}]}
\newcommand{\ruleLSend}{\ruleL{Send}}
\newcommand{\ruleLRecv}{\ruleL{Recv}}
\newcommand{\ruleLEps}{\ruleL{Eps}}

\newcommand{\ruleWf}[1]{[\rulename{WF-#1}]}
\newcommand{\ruleWfRty}{\ruleWf{Rty}}

\newcommand{\ruleP}[1]{[\rulename{P-#1}]}
\newcommand{\rulePSend}{\ruleP{Send}}
\newcommand{\rulePRecv}{\ruleP{Recv}}
\newcommand{\rulePPhi}{\ruleP{Phi}}
\newcommand{\rulePRecOne}{\ruleP{Rec-In}}
\newcommand{\rulePRecTwo}{\ruleP{Rec-Out}}
\newcommand{\rulePVar}{\ruleP{Var}}
\newcommand{\rulePEnd}{\ruleP{End}}

\newcommand{\rulePEmpty}{\ruleP{Empty}}
\newcommand{\rulePVarOmega}{\ruleP{Var-$\omega$}}
\newcommand{\rulePVarZero}{\ruleP{Var-$0$}}

\newcommand{\scrib}{\textsc{Scribble}\xspace}

\newcommand{\Q}{\mathbb{Q}}
\newcommand{\Rarr}{\Rightarrow}
\newcommand{\rarr}{\rightarrow}
\newcommand{\Larr}{\Leftarrow}

\newcommand{\keyword}[1]{{\color{keywordcolour}\code{#1}}}

\newcommand{\enc}[1]{\llbracket #1 \rrbracket}
\newcommand{\fv}[1]{\operatorname{fv}\!\left({#1}\right)}
\newcommand{\valid}[1]{\operatorname{\textsf{\upshape Valid}}\!\left({#1}\right)}
\newcommand{\subj}[1]{\operatorname{subj}(#1)}

\newcommand{\ang}[1]{\langle #1 \rangle}

\newcommand{\setof}[1]{\{#1\}}

\newcommand{\fsharp}{F\#\xspace}
\newcommand{\fstar}{\texorpdfstring{\textsc{F}$^{\star}$}{FStar}\xspace}

\newcommand{\shaded}[1]{\colorbox{gray!30}{$#1$}}

\newcommand{\eoft}{\mathrel{:}}
\newcommand{\emptyctx}{{\ensuremath{\varnothing}}}
\newcommand{\ctxc}[3]{#1,\dexp{#2}\eoft\dte{#3}} 

\newcommand{\withcolor}[2]{\colorlet{currbkp}{.}\color{#1}{#2}\color{currbkp}}
\newcommand{\colorse}{Teal} 
\newcommand{\colorlbl}{Indigo} 
\newcommand{\colorexp}{Blue} 
\newcommand{\colorte}{DarkOrchid} 
\newcommand{\colortp}{NavyBlue} 
\newcommand{\colorgp}{VioletRed} 

\newcommand{\dte}[1]{\withcolor{\colorte}{#1}} 

\newcommand{\tbool}{\dte{\code{bool}}}
\newcommand{\tint}{\dte{\code{int}}}
\newcommand{\tstr}{\dte{\code{string}}}
\newcommand{\tunit}{\dte{\code{unit}}}

\newcommand{\ppt}[1]{\withcolor{\colorse}{\textbf{\textsf{\upshape #1}}}} 

\newcommand{\dlbl}[1]{\withcolor{\colorlbl}{#1}} 

\newcommand{\gtctx}[2]{\ang{{#1} \prec \dgt{#2}}}
\newcommand{\ltctx}[2]{\ang{{#1} \prec \dtp{#2}}}
\newcommand{\ltequiv}{\stepsto[\epsilon]}
\newcommand{\ltequivmany}{\stepsto[\epsilon]^*}
\newcommand{\ltassoc}{\Leftrightarrow}

\newcommand{\ctxext}[3]{#1 \cup \setof{\dexp{#2}: \dte{#3}}}

\newcommand{\Sigmap}{\Sigma_{\ppt p}}
\newcommand{\Sigmaq}{\Sigma_{\ppt q}}
\newcommand{\Sigmar}{\Sigma_{\ppt r}}
\newcommand{\Sigmas}{\Sigma_{\ppt s}}

\newcommand{\Lp}{\dtp{L}_{\ppt p}}
\newcommand{\Lq}{\dtp{L}_{\ppt q}}
\newcommand{\Lr}{\dtp{L}_{\ppt r}}
\newcommand{\Ls}{\dtp{L}_{\ppt s}}
\newcommand{\Lt}{\dtp{L}_{\ppt t}}
\newcommand{\LA}{\dtp{L}_{\ppt A}}
\newcommand{\LB}{\dtp{L}_{\ppt B}}
\newcommand{\LC}{\dtp{L}_{\ppt C}}

\newcommand{\dgt}[1]{\withcolor{\colorgp}{#1}}

\newcommand{\gtbran}[3]{\dgt{\ppt{#1} \to \ppt{#2}\left\{#3\right\}}}
\newcommand{\gtbransingle}[3]{\dgt{\ppt{#1} \to \ppt{#2}:#3}}
\newcommand{\gtvar}[2]{\dgt{\mathbf{#1}\ang{#2}}}
\newcommand{\gtrecur}[4]{\dgt{\mu\mathbf{#1}(#2)\ang{#3}.#4}}
\newcommand{\gtrecursimpl}[3]{\dgt{\mu\mathbf{#1}(#2).#3}}
\newcommand{\gtend}{\dgt{\code{end}}}

\newcommand{\gtctxproj}[3]{\gtctx{#1}{#2} \upharpoonright \ppt{#3}}
\newcommand{\ctxproj}[2]{#1 \upharpoonright \ppt{#2}}

\newcommand{\dtp}[1]{\withcolor{\colortp}{#1}} 

\newcommand{\tend}{\dtp{\code{end}}}

\newcommand{\tgeneric}[2]{\dtp{\ppt{#1}\dagger\mspace{-3mu}\left\{#2\right\}}}
\newcommand{\toffer}[2]{\dtp{\ppt{#1}\&\mspace{-3mu}\left\{#2\right\}}}
\newcommand{\toffersingle}[2]{\dtp{\ppt{#1}\&#2}}
\newcommand{\ttake}[2]{\dtp{\ppt{#1}\mspace{-4mu}\oplus\mspace{-4mu}\left\{#2\right\}}}
\newcommand{\ttakesingle}[2]{\dtp{\ppt{#1}\mspace{-4mu}\oplus\mspace{-4mu}#2}}
\newcommand{\tvar}[2]{\dtp{\mathbf{#1}\ang{#2}}}
\newcommand{\trecur}[5]{\dtp{\mu\mathbf{#1}\,#2(#3)\ang{#4}.#5}}
\newcommand{\trecursimpl}[4]{\dtp{\mu\mathbf{#1}\,#2(#3).#4}}
\newcommand{\tphi}[4]{\dtp{\dlbl{#1}(\vb{#2}{#3}). #4}}

\newcommand{\dexp}[1]{\withcolor{\colorexp}{#1}} 

\newcommand{\vb}[2]{\dexp{#1}\mathop{:} \dte{#2}}

\newcommand{\etrue}{\dexp{\code{true}}}
\newcommand{\efalse}{\dexp{\code{false}}}

\newcommand{\eintlit}[1]{\dexp{#1}}

\newcommand{\stepsto}[1][\quad]{\xrightarrow{#1}} 
\newcommand{\stepstomany}[1][\quad]{\stepsto{#1}^*} 
\newcommand{\subst}[2]{[{#1}/{#2}]} 

\newcommand{\esetof}[1]{\setof{\dexp{#1}}}
\newcommand{\psetof}[1]{\setof{\ppt{#1}}}
\newcommand{\bigP}{\ppt{$\mathbb{P}$}}

\newcommand{\red}[1]{{\color{red} #1}}

\newcommand{\bnfas}{\mathrel{::=}}
\newcommand{\bnfalt}{\mathrel{\mid}}

\newcommand{\ltsmsg}[5]{\ppt{#1} \to \ppt{#2} : \dlbl{#3}(\dexp{#4}: \dte{#5})}
\newcommand{\ltsmsgsimpl}[3]{\ppt{#1} \to \ppt{#2} : \dlbl{#3}}

\tikzstyle{block} = [rectangle, draw, fill=blue!20,
    text width=12em, text centered, rounded corners, minimum height=4em]
\tikzstyle{line} = [draw, -latex']

\newcommand{\tikzmark}[1]{\tikz[overlay,remember picture] \node (#1) {};}

\newcommand*{\AddNote}[4]{%
    \begin{tikzpicture}[overlay, remember picture]
        \draw [decoration={brace,amplitude=0.5em},decorate,thick]
            ($(#3)!([yshift=1.5ex]#1)!($(#3)-(0,1)$)$) --
            ($(#3)!(#2)!($(#3)-(0,1)$)$)
                node [align=center, text width=1.5cm, pos=0.5, anchor=west] {#4};
    \end{tikzpicture}
}%
\newcommand*{\AddNoteLeft}[4]{%
    \begin{tikzpicture}[overlay, remember picture]
        \draw [decoration={brace,mirror,amplitude=0.5em},decorate,thick]
            ($(#3)!([yshift=1.5ex]#1)!($(#3)-(0,1)$)$) --
            ($(#3)!(#2)!($(#3)-(0,1)$)$)
                node [align=center, text width=1.5cm, pos=0.5, anchor=east] {#4};
    \end{tikzpicture}
}%

\newcommand{\mathcolorbox}[2]{\colorbox{#1}{$\displaystyle #2$}}

%% file: abstract.tex
\begin{abstract}
With distributed computing becoming ubiquitous in the modern era, safe
distributed programming is an open challenge.
To address this, multiparty session types (MPST) provide a typing discipline
for message-passing concurrency, guaranteeing
communication safety properties such as deadlock freedom.

While originally MPST focus on the communication aspects, and employ a
simple typing system for communication payloads, communication protocols in the
real world usually contain \emph{constraints} on the payload.
We introduce \emph{refined multiparty session types (RMPST)}, an extension of
MPST,
that express data dependent protocols via \emph{refinement types} on the
data types.

We provide an implementation of RMPST, in a toolchain called \Ourtool, using
\scrib, a multiparty protocol description toolchain, and targeting \fstar,
a verification-oriented functional programming language.
Users can describe a protocol in \scrib
and implement the endpoints in \fstar using \emph{refinement-typed} APIs
generated from the protocol.
The \fstar compiler can then statically verify the refinements.
Moreover, we use a novel approach of callback-styled API generation, providing
\emph{static} linearity guarantees with the inversion of control.
We evaluate our approach with real world examples and show that it has little
overhead compared to a na\"{i}ve implementation, while guaranteeing safety
properties from the underlying theory.

\keywords{
    Multiparty Session Types (MPST) \and
    Refinement Types \and
    Code Generation \and
    \fstar}
\end{abstract}

%% file: intro.tex
\section{Introduction}

Distributed interactions and message passing are fundamental elements of the
modern computing landscape.
Unfortunately, language features and support for high-level and \emph{safe}
programming of communication-oriented and distributed programs are much
lacking,
in comparison to those enjoyed for more traditional ``localised'' models of
computation.
One of the research directions towards addressing this challenge is
\emph{concurrent behavioural types}~\cite{BettyBook, FTPL16BehavioralSurvey},
which seek to extend the benefits of conventional
type systems, as the most successful form of lightweight formal
verification, to communication and concurrency.

\emph{Multiparty session types} (MPST)~\cite{POPL08MPST, JACM16MPST}, one of
the most active
topics in this area, offer a theoretical framework for specifying message passing
protocols between multiple participants.
MPST use a type system--based approach to
\emph{statically} verify whether a system of processes implements a given
protocol safely.
The type system guarantees key execution properties such as freedom from
message reception errors or deadlocks.
However, despite much recent progress, there remain large gaps between
the current state of the art and (i) powerful and \emph{practical}
languages and techniques available to programmers today, and (ii) more
advanced type disciplines needed to express a wider variety of constraints of
interaction found in real-world protocols.

This paper presents and combines two main developments: a theory of MPST
enriched with \emph{refinement types} \cite{PLDI91RefinementML}, and a
practical method, \emph{callback-based programming}, for safe session
programming.
The key ideas are as follows:

\mypara{\emph{Refined} Multiparty Session Types (RMPST)}
The MPST theory~\cite{POPL08MPST, JACM16MPST} provides a core framework for
decomposing (or \emph{projecting}) a \emph{global type} structure, describing
the collective behaviour of a distributed system, into a set of
participant-specific \emph{local types} (see \cref{fig:top-down-mpst}).
The local types are then used to implement endpoints.

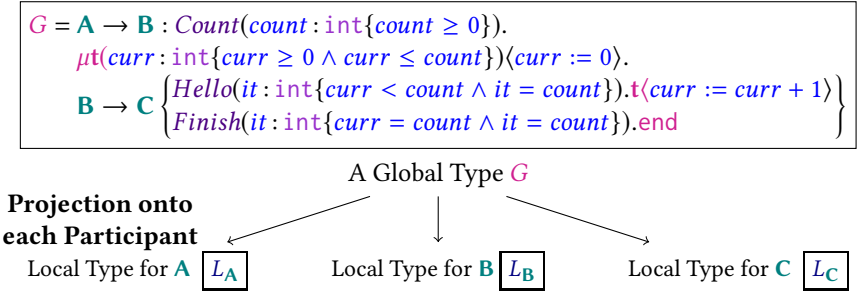
\begin{figure}
  \begin{tikzpicture}
    \node[draw, rectangle] (G) {
    $ \arraycolsep=1pt
      \begin{array}[t]{@{}ll@{}}
        \dgt{G} = &
         \gtbransingle{A}{B}{\dlbl{Count}(\vb{count}{\tint\esetof{count \geq 0}})}. \\
         & \gtrecur{t}
         {\vb{curr}{\tint\esetof{curr \geq 0 \land curr \leq count}}}
         {\dexp{curr := 0}}
         {}\\
           &\begin{array}{@{}l@{}}
              \gtbran{B}{C}{
                \begin{array}{@{}l@{}}
                  \dlbl{Hello}(
                  \vb{it}{\tint\esetof{curr < count \land it = count})}
                .\gtvar{t}{\dexp{curr:=curr+1}}\\
                  \dlbl{Finish}(\vb{it}{\tint}\esetof{curr = count \land it = count}).\gtend
                \end{array}
              }
           \end{array}
      \end{array}$};
    \node[below=0.5mm of G] (Gtext) {A Global Type $\dgt G$};
    \node[below=-1.5mm of Gtext, xshift=-45mm, align=center] (proj) {\bf
      Projection onto \\ \bf each
      Participant};
    \node[below=6mm of Gtext, xshift=-40mm] (LA) {\small Local Type for \ppt A~
      \boxed{\LA}};
    \node[below=6mm of Gtext] (LB) {\small Local Type for \ppt B~\boxed{\LB}};
    \node[below=6mm of Gtext, xshift=40mm] (LC) {\small Local Type for \ppt C~
      \boxed{\LC}};
    \draw[->] (Gtext) -- (LA);
    \draw[->] (Gtext) -- (LB);
    \draw[->] (Gtext) -- (LC);
  \end{tikzpicture}
  \vspace{-3mm}
  \caption{Top-down View of (R)MPST}
  \vspace{-3mm}
  \label{fig:top-down-mpst}
\end{figure}

Our theory of RMPST follows the same top-down methodology, but
enriches MPST with features from \emph{refinement types}~\cite{PLDI91RefinementML},
to support the elaboration of data types in global and local types.
Refinement types allow \emph{refinements} in the form
of logical predicates and constraints to be attached to a base type.
This allows to express various constraints in protocols.

To motivate our refinement type extension, we use a counting protocol shown in
\cref{fig:top-down-mpst}, and leave the details to \cref{section:theory}.
Participant $\ppt A$ sends $\ppt B$ a number with a $\dlbl{Count}$ message.
In this message, the refinement type $\vb{count}{\tint\esetof{count \geq 0}}$
restricts the value for $\dexp{count}$ to be a natural number.
Then $\ppt B$ sends $\ppt C$ exactly that number of $\dlbl{Hello}$ messages,
followed by a $\dlbl{Finish}$ message.

We demonstrate how refinement types are used to \emph{better} specify the protocol:
The counting part of the protocol is described using a recursive type with two
branches, where we use refinement types to restrict the protocol flow.
The variable $\dexp{curr}$ is a \emph{recursion variable}, which remembers the current
iteration, initialised to $\dexp 0$, and increments on each recursion ($\dexp{curr :=
curr + 1}$).
The refinement $\dexp{curr = count}$ in the $\dlbl{Finish}$ branch specifies
that the branch may only be taken at the last iteration;
the refinement $\dexp{it = count}$ in both $\dlbl{Hello}$ and
$\dlbl{Finish}$ branches specifies a payload value \emph{dependent} on the recursion
variable $\dexp{curr}$ and the variable $\dexp{count}$ transmitted in the first message.

We establish the correctness of Refined MPST.
In particular, we show that projection is behaviour-preserving and that
well-formed global types with refinements satisfy progress, i.e.\ they do not get stuck.
Therefore, if the endpoints follow the behaviour prescribed by the local types,
derived (via projection) from a well-formed global type with refinements, the
system is deadlock-free.

\mypara{\emph{Callback-styled, Refinement-typed} APIs for Endpoint
  Implementations}
One of the main challenges in applying session types in practice
is dealing with session \emph{linearity}: a session
channel is used \emph{once and only once}.
Session typing relies on a linear treatment of communication channels, in order to
 track the I/O actions performed on the channel against the intended
session type.
Most existing implementations adopt one of two approaches: monadic interfaces
in functional
languages~\cite{HaskellSessionBookChapter,
  SCP18OCaml, ECOOP20OCamlMPST},
or ``hybrid'' approaches that complement
static typing with dynamic linearity checks~\cite{FASE16EndpointAPI,
  scalas17linear}.

This paper proposes a fresh look to session-based programming that does
\emph{not} require linearity checking for static safety.
We promote a form of session programming where session I/O
is \emph{implicitly} implemented by \emph{callback functions}
--- we say ``implicitly'' because the user does not write any I/O operations
themself: an \emph{input callback} is written to \emph{take} a received message
as a parameter, and an \emph{output callback} is written to simply
\emph{return} the label and value of the message to be sent.

The callbacks are supported by a runtime, generated along with APIs in refinement types
according to the local type.
The runtime performs communication and invokes user-specified callback
functions upon corresponding communication events.
We provide a code generation tool to streamline the writing of callback
functions for the projected local type.

The inversion of control allows us to dispense with linearity
checking, because our approach does not expose communication channels to the user.
Our approach is a natural fit to functional programming settings, but also directly
applicable to any statically typed language.
Moreover, the linearity guarantee is achieved \emph{statically} without the use
of a linear type system, a feature that is usually not supported by mainstream
programming languages.
We follow the principle of event-based programming via the
use of callbacks, prevalent in modern days of computing.

\mypara{A Toolchain Implementation: \Ourtool}
To evaluate our proposal, we implement RMPST with a toolchain ---
\Ourtool, as an extension to the
\scrib toolchain~\cite{ScribbleWebsite, ScribbleBookChapter}
(\url{http://www.scribble.org/}) with \fstar~\cite{POPL16FStar} as the target endpoint
language.

Building on our callback approach, we show how to integrate RMPST
with the verification-oriented functional programming language \fstar,
exploiting its capabilities of refinement types and static verification to
extend our fully static safety guarantees to data refinements in sessions.
Our experience of specifying and implementing protocols drawn from the
literature and real-world applications attests to the practicality of our
approach and the value of statically verified refinements.
Our integration of RMPST and \fstar allows developers to utilise advanced type system features
to implement safe distributed application protocols.

\mypara{Paper Summary and Contributions}

\begin{itemize}[leftmargin=*]
\item[\cref{section:overview}]
  We present an overview of our toolchain: \Ourtool, and provide background
  knowledge of \scrib and MPST.
  We use a number guessing game, \lstinline{HigherLower}, as our running
  example.

\item[\cref{section:implementation}]
  We introduce \Ourtool, a toolchain for RMPST.
  We describe in detail how our generated APIs can be used to implement
  multiparty protocols with refinements.

\item[\cref{section:theory}]
We establish the metatheory of RMPST, which gives semantics of
global and local types with refinements.
We prove trace equivalence of global and local types w.r.t.\ projection
(\cref{thm:trace-eq}), and show progress and type safety of well-formed global
types (\cref{thm:progress} and \cref{thm:type-safety}).

\item[\cref{section:evaluation}]
  We evaluate our toolchain with examples from the session type literature, and
  measure the time taken for compilation and execution.
  We show that our toolchain does not have a long compilation time, and our
  runtime does not incur a large overhead on execution time.
\end{itemize}

We submitted an artifact for evaluation~\cite{sessionstarartifact}, containing the source code of our
toolchain \Ourtool, with examples and benchmarks used in the evaluation.
The artifact is available as a Docker image, and can be
\href{https://hub.docker.com/layers/sessionstar2020/sessionstar/artifact/images/sha256-4e3bf61238e04c1d2b5854971f0ef78f0733d566bf529a01d9c3b93ffa831193?context=explore}{accessed}
on the Docker Hub.
The source files are available on GitHub
(\url{https://github.com/sessionstar/oopsla20-artifact}).
We present the proof of our theorems, and additional technical details of the
toolchain, in the \iftoggle{fullversion}{appendix.}{full version of the paper
  (to be added in the camera-ready version).}


%% file: overview.tex
\section{Overview of Refined Multiparty Session Types}
\label{section:overview}

In this section, we give an overview of our toolchain: \Ourtool,
describing its key toolchain stages.
\Ourtool extends the \scrib toolchain with refinement types and uses \fstar as
a target language.
We begin with a short background on basic multiparty session types and \scrib,
then demonstrate the specification of
distributed applications with refinements
using the extended \scrib.

\subsection{Toolchain Overview}
\label{subsection:overview-overview}

\input{fig/fig-workflow-detail.tex}
We present an overview of our toolchain in \cref{fig:workflow-detail}, where we
distinguish user provided input by developers in \boxed{\text{solid boxes}},
from generated code or toolchain internals in \dbox{dashed boxes}.

Development begins with \emph{specifying a protocol} using an extended \scrib
protocol description language.
\scrib is closely associated with the MPST theory \cite{ScribbleBookChapter,
  FeatherweightScribble}, and provides a user-friendly syntax for multiparty
protocols.
We extend the \scrib toolchain to support RMPST, allowing refinements to be
added via annotations.
The extended \scrib toolchain (as part of \Ourtool) validates the
well-formedness of the protocol, and produces a representation in the form of a
\emph{communicating finite state machine} (CFSM, \cite{JACM83CFSM}) for a given
participant.

We then use a code generator (also as part of \Ourtool) to generate \fstar APIs
from the CFSM, utilising a number of advanced type system features available in
\fstar (explained later in \cref{subsection:fstar-bg}).
The generated APIs, described in detail in
\cref{section:implementation}, consist of various type definitions, and an entry point
function taking \emph{callbacks} and \emph{connections} as arguments.

In our design methodology, we separate the concern of \emph{communication} and
\emph{program logic}.
The callbacks, corresponding to program logic, do not involve
communication primitives ---
they are invoked to prompt a value to be sent, or to process a received value.
Separately, developers provide a connection that allows base types to be
serialised and transmitted to other participants.
Developers implement the endpoint by providing both \emph{callbacks} and
\emph{connections}, according to the generated refinement typed APIs.
They can run the protocol by invoking the generated entry point.
Finally, the \fstar source files can be verified using the \fstar compiler, and
extracted to an OCaml program (or other supported targets) for efficient
execution.

\subsection{Global Protocol Specification --- RMPST in Extended \scrib}
\label{subsection:impl-rmpst}

The workflow in the standard MPST theory~\cite{POPL08MPST}, as is generally the case
in distributed application development, starts from identifying the intended
  protocol for participant interactions.
In our toolchain, a \emph{global protocol}---the description of the whole
protocol between all participants from a bird eye's view---is specified
using our RMPST extension of the \scrib protocol description
language~\cite{ScribbleWebsite, ScribbleBookChapter}.
\Cref{fig:guess} gives the global protocol for a three-party game,
\lstinline+HigherLower+, which we use as a running example.

\mypara{Basic \scrib/MPST}

We first explain basic \scrib (corresponding to the standard MPST) \emph{without}
the \lstinline+@+-annotations (annotations are extensions to the basic \scrib).
\begin{enumerate}[leftmargin=*]
\item
  The main protocol \lstinline+HigherLower+ declares three \emph{roles}
  $\ppt A$, $\ppt B$ and $\ppt C$, representing the runtime
  communication session participants.
  The protocol starts with $\ppt A$ sending $\ppt B$ a
  \lstinline+start+ message and a \lstinline+limit+ message, each carrying an
  \keyword{\small int} payload.

\item
  The \keyword{\small do} construct specifies all three roles to proceed according to the
  (recursive) \lstinline+Aux+ sub-protocol.
  $\ppt C$  sends $\ppt B$ a \lstinline+guess+ message, also carrying
  an \keyword{\small int}.
  (The \keyword{\small aux} keyword simply tells \scrib that a sub-protocol does not
  need to be verified as a top-level entry protocol.)

\item
  The \keyword{\small choice at}\code{\small~B} construct specifies at this point that
  $\ppt B$ should
  decide (make an \emph{internal} choice) by which one of the four cases the
  protocol should proceed.
  This decision is explicitly communicated (as an \emph{external} choice) to
  $\ppt A$ and $\ppt C$ via the messages in each case.
  The \lstinline+higher+ and \lstinline+lower+ cases are the recursive cases,
  leading to another round of \lstinline+Aux+ (i.e.\ another \lstinline+guess+
  by $\ppt C$); the other two cases, \lstinline+win+ and \lstinline+lose+,
  end the protocol.
\end{enumerate}

\noindent
To sum up, $\ppt A$ sends $\ppt B$ two numbers, and $\ppt C$
sends a number (at least one) to $\ppt B$ for as long as $\ppt B$
replies with either \lstinline+higher+ or \lstinline+lower+ to $\ppt C$
(and $\ppt A$).
Next we demonstrate how we can express data
dependencies using refinements with our extended \scrib.

\begin{figure}[t]
    \centering
    \begin{minipage}{0.9\textwidth}
\begin{lstlisting}[language=Scribble,basicstyle={\small\ttfamily}]
global protocol HigherLower(role A, role B, role C) {
  // A tells B a secret number `n0`,
  // and the number `t0` of attempts that C has to guess it
  start(n0:int) from A to B;   @'0<=n0<100' £\label{line:start}£
  limit(t0:int) from A to B;   @'0<t0'
  do Aux(A, B, C);             @'B[n0, t0]' £\label{line:aux-init}£ }
aux global protocol Aux(role A, role B, role C) @'B[n:int{0<=n<100}, t:int{0<t}]' { £\label{line:aux-def}£
  guess(x:int)             from C to B;      @'0<=x<100'         // Next guess by C
  choice at B {   higher() from B to C;      @'n>x && t>1'       // Secret is higher £\label{line:higher}£
                  higher() from B to A;
                  do Aux(A, B, C);           @'B[n, t-1]' £\label{line:aux-call}£
         } or {   win()    from B to C;      @'n=x'              // C wins, A loses £\label{line:win}£
                  lose()   from B to A;
         } or {   lower()  from B to C;      @'n<x && t>1'       // Secret is lower £\label{line:lower}£
                  lower()  from B to A;
                  do Aux(A, B, C);           @'B[n, t-1]'
         } or {   lose()   from B to C;      @'n!=x && t=1'       // A wins, C loses £\label{line:lose}£
                  win()    from B to A;
}        }
\end{lstlisting}
    \end{minipage}
    \vspace{-3mm}
    \caption{A \emph{Refined} \scrib Global Protocol for a \lstinline+HigherLower+ Game. } 
    \vspace{-3mm}
    \label{fig:guess}
\end{figure}

\input{extended-scribble}

%% file: fig/fig-workflow-detail.tex
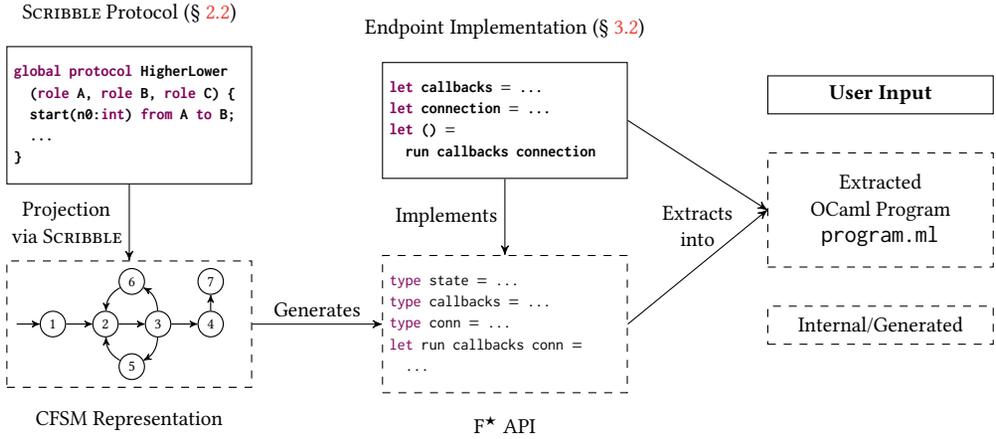
\begin{figure}[t]
  \centering
  \footnotesize
  \begin{tikzpicture}[
      >=stealth',
    box/.style={draw, rectangle, minimum width=3cm, minimum height=15mm, node
      distance=50mm, align=center}
  ]
    \node[box, dashed] (cfsm) at (0, 0) {
      \begin{minipage}{3cm}
        \scalebox{0.7}{
        \hspace{-3mm}
      \begin{tikzpicture}[->,>=stealth',shorten >=0.01pt,auto,node distance=1mm,minimum width=2mm,minimum height=2mm,
          semithick, solid]
          \draw node[] (start) at (-0.8, 0) {};
          \draw node[circle, draw] (1) at (0, 0) {1};
          \draw node[circle, draw] (2) at (1, 0) {2};
          \draw node[circle, draw] (3) at (2, 0) {3};
          \draw node[circle, draw] (4) at (3, 0) {4};
          \draw node[circle, draw] (5) at (1.5, -0.8) {5};
          \draw node[circle, draw] (6) at (1.5, 0.8) {6};
          \draw node[circle, draw] (7) at (3, 0.8) {7};
          \path (start) edge node {} (1);
          \path (1) edge node {} (2);
          \path (2) edge node {} (3);
          \path (3) edge node {} (4);
          \path (3) edge [bend left, right] node {} (5);
          \path (3) edge [bend right, right] node {} (6);
          \path (4) edge node {} (7);
          \path (5) edge [bend left] node {} (2);
          \path (6) edge [left, bend right] node {} (2);
      \end{tikzpicture}}
      \end{minipage}
    };
    \node[box, above=1cm of cfsm] (scribble) {
      \begin{minipage}{3cm}
        \begin{lstlisting}[language=Scribble,numbers=none,basicstyle=\tiny\bf\ttfamily]
global protocol HigherLower
  (role A, role B, role C) {
  start(n0:int) from A to B;
  ...
}
\end{lstlisting}
      \end{minipage}
    };
    \node[box, right of=cfsm, dashed] (fstar-api) {
      \begin{minipage}{3cm}
        \begin{lstlisting}[language=FSharp,numbers=none,basicstyle=\tiny\ttfamily]
type state = ...
type callbacks = ...
type conn = ...
let run callbacks conn =
  ...
\end{lstlisting}
      \end{minipage}

    };
    \node[box, above=1cm of fstar-api] (fstar-impl) {
      \begin{minipage}{3cm}
        \begin{lstlisting}[language=FSharp,numbers=none,basicstyle=\tiny\bf\ttfamily]
let callbacks = ...
let connection = ...
let () =
  run callbacks connection
\end{lstlisting}
      \end{minipage}
    };
    \node[box, right of=fstar-api, yshift=1.5cm, dashed] (ocaml) {Extracted \\ OCaml Program \\ \small\code{program.ml}};
    \draw[->] (scribble.south) -- node[left,align=center] {Projection \\ via \scrib} (cfsm.north);
    \draw[->] (cfsm.east) -- node[above] {Generates} (fstar-api.west);
    \draw[->] (fstar-impl.south) -- node[left] {Implements} (fstar-api.north);
    \draw[->] (fstar-api.east) -- node[above] {} (ocaml.west);
    \draw[->] (fstar-impl.east) -- node[yshift=-8mm,align=center] {Extracts \\ into} (ocaml.west);
    \node[above=2mm of scribble] {\scrib Protocol (\cref{subsection:impl-rmpst})};
    \node[above=2mm of fstar-impl] {Endpoint Implementation (\cref{subsection:impl-fstar-api})};
    \node[below=2mm of cfsm] {CFSM Representation};
    \node[below=2mm of fstar-api] {\fstar API};

    \node[draw, rectangle, minimum width=3cm, minimum height=0.5cm, above=0.5cm
    of ocaml] {\textbf{User Input}};
    \node[draw, rectangle, minimum width=3cm, minimum height=0.5cm, below=0.5cm of ocaml, dashed] {Internal/Generated};
  \end{tikzpicture}
  \vspace{-1mm}
  \caption{Overview of Toolchain}
  \label{fig:workflow-detail}
  \vspace{-3mm}
\end{figure}

%% file: extended-scribble.tex
\mypara{Extended \scrib/RMPST}

As described above, a basic global protocol (equivalent to a standard MPST
\emph{global type}) specifies the structure of participant interactions, but is
not very informative about the behaviour of the \emph{application} in the
aspect of data transmitted.
This limitation can be contrasted against standardised (but informal)
specifications of real-world application protocols, e.g.\ HTTP, OAuth, where a
significant part, if not the majority, of the specification is devoted to the
\emph{data} side of the protocol.
It goes without saying, details missing from the specification cannot be
verified on an implementation.

We go through \cref{fig:guess} again, this time including the practical
extensions proposed by this paper to address these limitations:
RMPST enables a refinement type--based treatment of protocols, capturing
and integrating both data constraints and interactions.
While refinement types by themselves can already greatly enrich the specification
of individual messages, the most valuable and interesting aspect is the
\emph{interplay} between data refinements and the consequent interactions
(i.e.\ the protocol flow) in the distributed, multiparty setting.
In our setup, we allow data-dependent refinements to be specified in the global
protocol,
we explain various ways of using them in the running example:

\newcommand{\lstinlinefst}{\lstinline[language=FSharp,basicstyle=\small\ttfamily]}
\begin{itemize}[leftmargin=*]
  \item \textbf{Message Values.}
  A basic use of refinements is on message \emph{values}, specifically their
  payload contents.
  The annotation on the first interaction (\cref{line:start}) specifies that
  $\ppt A$ and $\ppt B$ not only communicate an \lstinlinefst+n0:int+,
  but that \lstinlinefst+0<=n0<100+ is a \emph{post}condition for the value
  produced by $\ppt A$ and a \emph{pre}condition on the value received by
  $\ppt B$.
  Similarly, the \lstinlinefst+int+ carried by \lstinline+limit+ must be positive.

  \item \textbf{Local Protocol State.}
  RMPST also supports refinements on the \emph{recursions} that a protocol
  transitions through.
  The \lstinline+B[...]+ annotation (\cref{line:aux-def}) in the
  \lstinline+Aux+ header specifies
  the \emph{local} state known by \lstinline+B+ during the recursion,
  whenever $\ppt B$ reaches \emph{this point} in the ongoing protocol.
  The local state includes an \keyword{\small int} \lstinline+n+
  such that \lstinlinefst+0<=n<100+ and an \keyword{\small int} \lstinline+t+ such that
  \lstinline+0<t+.
  These extra variables are available for all enactments of this subprotocol.
  That is, on the first entry from \lstinline+HigherLower+, where the
  \lstinline+do+ annotation \lstinline+B[n0, t0]+ (\cref{line:aux-init})
  specifies the initial values; and from the recursive entries
  (\cref{line:aux-call}).

  By \emph{known} state, we mean that $\ppt B$ will have access to the
  exact values at runtime, although statically we can only be sure that they
  lie within the intervals specified in the refinements.
  Other session participants can only use the type information,
  without knowing the value,
  e.g.\ $\ppt C$ does not know the exact value of \lstinline+n+ (which is
  the main point of this game), but knows the range via the refinements, and
  hence the endpoint may utilise this knowledge for reasoning.

  \item \textbf{Protocol Flow.}
  As mentioned, RMPST combines protocol specifications with
  refinements in order to direct the flow of the
  protocol --- specifically at internal choices.
  The annotation on the \lstinline+win+ interaction (\cref{line:win}) from \lstinline+B+ to
  $\ppt C$ specifies that \lstinline+B+ can only send this message, and
  thus select this \keyword{\small choice} case, after a correct guess by
  $\ppt C$.
  Similarly, $\ppt B$ can only select the other cases after an incorrect
  guess: \lstinline+lose+ (\cref{line:lose}) when $\ppt C$ is on its last attempt, or the
  corresponding \lstinline+higher+ (\cref{line:higher}) or \lstinline+lower+
  (\cref{line:lower}) cases otherwise.
  We exploit the fact that a refinement type can be uninhabited due to the
  impossibility to satisfy the constraint on the type, to encode protocol flow
  conditions.
\end{itemize}

Refinements allow a basic description of an interaction structure
to be elaborated into a more expressive application specification.
Note the \lstinlinefst+t>1+ in the \lstinline+higher+ and \lstinline+lower+
refinements are necessary to ensure that the \lstinlinefst+0<t+ in the
\lstinline+Aux+ refinement is satisfied, given that the \keyword{\small do}
annotations specify \lstinline+Aux+ to be recursively enacted with
\lstinlinefst+t-1+.
Albeit simple, the protocol shows how we can use refinements in various
means to express data and control flow dependencies and recursive invariants in
a multiparty setup.
Once the protocol is specified, our toolchain allows the refinements
 in RMPST to be directly mapped
to data refinements in F* through a callback-styled API generation.

%% file: implementation.tex
\section{Implementing Refined Protocols in \fstar}
\label{section:implementation}

In this section, we demonstrate our callback-styled, refinement-typed APIs
for implementing endpoints in \fstar \cite{POPL16FStar}.
We introduced earlier the workflow of our toolchain
(\cref{subsection:overview-overview}).
In \cref{subsection:fstar-bg}, we
summarise the key features of \fstar utilised in our implementation.
Using our running example, we explain the generated APIs in \fstar in
\cref{subsection:impl-fstar-api},
and how to implement endpoints in
\cref{subsection:impl-fstar-impl}.
We outline the function we generate for executing the endpoint in
\cref{subsection:impl-run-fsm}.

Developers using \Ourtool implement the callbacks in \fstar,
to utilise the functionality of refinement types provided by the \fstar
type-checker.
The \fstar compiler can verify statically whether the provided implementation
by the developer satisfies the refinements as specified in the protocol.
The verified code can be extracted via the \fstar compiler to OCaml and
executed.

The verified implementation enjoys properties both from the MPST theory, such
as \emph{session fidelity} and \emph{deadlock freedom}, but also from
refinement types, that the data dependencies are verified to be correct
according to the protocol.
Additional details on code generation can be found in
\iftoggle{fullversion}{\cref{section:impl-appendix}}{the full version of the
  paper}.

\input{fstar-background}

\subsection{Projection and \fstar API Generation -- Communicating Finite State
  Machine--based Callbacks for Session I/O}
\label{subsection:impl-fstar-api}

As in the standard MPST workflow, the next step (\cref{fig:workflow-detail}) is to
\emph{project} our \emph{refined} global protocol onto each role.
This decomposes the global protocol into a set of \emph{local} protocols,
prescribing the view of the protocol from each role.
Projection is the definitive mechanism in MPST: although all endpoints together
must comply to global protocol, projection allows each endpoint to be
\emph{separately} implemented and verified, a key requirement for practical distributed
programming.
As we shall see, the way projection treats refinements---we must
consider the local \emph{knowledge} of values propagated through the
multiparty protocol---is crucial to verifying refined implementations,
including our simple running example.

\mypara{Projection onto \textsf{\textbf{\color{Teal}B}}}

\begin{figure}[t]
\small
\begin{subfigure}{\textwidth}
\begin{tikzpicture}
[
  >=stealth',
  NODE/.style={draw,circle,minimum width=3mm,inner sep=0pt,fill=lightgray}
]
\node[NODE] (B1) {$1$};
\node[NODE,right=18mm of B1] (B2) {$2$};
\node[NODE,right=18mm of B2,label={[align=left,xshift=-13pt,yshift=-4pt]above:{$\mathcolorbox{lightgray}{n\{0 \mathop{\leq} n \mathop{<} 100\},}$\\[-3pt]$\mathcolorbox{lightgray}{t\{0 \mathop{<} t\}}$}}] (B3) {$3$};
\node[NODE,right=18mm of B3,yshift=7mm] (B5) {$5$};
\node[NODE,right=18mm of B3,yshift=-7mm] (B6) {$6$};
\node[NODE,right=18mm of B5,yshift=-7mm] (B4) {$4$};
\node[NODE,right=18mm of B4,yshift=7mm] (B7) {$7$};
\node[NODE,right=18mm of B4,yshift=-7mm] (B8) {$8$};
\node[NODE,right=18mm of B7,yshift=-7mm] (B9) {};
\draw[->] ($(B1.west)-(3mm,0)$) -- (B1);
\draw[->] (B1) -- node [above,xshift=-10pt,yshift=4pt]{$A?\mathtt{start(n_0)\{0 \mathop{\leq} n_0 \mathop{<} 100\}}$} (B2);
\draw[->] (B2) -- node [below,xshift=-4pt,yshift=-4pt]{$A?\mathtt{limit(t_0)\{0 \mathop{\leq} t_0\}}$} (B3);
\draw[->] (B3) -- node [above,yshift=-1pt]{$C?\mathtt{guess(x)\{0 \mathop{\leq} x \mathop{<} 100\}}$} (B4);
\draw[>=angle 60,->,densely dotted] ($(B3.north west)+(-2.5mm,4.5mm)$) -- (B3.north west);
\path[->] (B4) edge[bend right=15] node [above,xshift=19pt,yshift=2pt]{$C!\mathtt{higher\{n \mathop{>} x \mathbin{\wedge} t \mathop{>} 1\}}$} (B5);
\path[->] (B5) edge[bend right=15] node [above,xshift=10pt,yshift=3pt]{$A!\mathtt{higher}$} (B3);
\path[->] (B4) edge[bend left=15] node [below,xshift=19pt,yshift=-2pt]{$C!\mathtt{lower\{n \mathop{<} x \mathbin{\wedge} t \mathop{>} 1\}}$} (B6);
\path[->] (B6) edge[bend left=15] node [below]{$A!\mathtt{lower}$} (B3);
\path[->] (B4) edge[bend left=15] node [below,xshift=15pt]{$C!\mathtt{win\{n \mathop{=} x\}}$} (B7);
\path[->] (B7) edge[bend left=15] node [above,yshift=1pt]{$A!\mathtt{lose}$} (B9);
\path[->] (B4) edge[bend right=15] node
[above,xshift=28pt,yshift=1pt]{$C!\mathtt{lose\{n \mathop{\neq} x \mathbin{\wedge} t \mathop{=} 1\}}$} (B8);
\path[->] (B8) edge[bend right=15] node [below,yshift=-1pt]{$A!\mathtt{win}$} (B9);
\end{tikzpicture}
\caption{CFSM Representation of the Projection.
$!$ stands for sending actions, and $?$ for
receiving actions on edges.}
\label{fig:higherlower_projection_B-cfsm}
\end{subfigure}

\smallskip

\renewcommand{\lstinlinefst}{\lstinline[language=FSharp,basicstyle=\small\ttfamily]}
{\small
\begin{tabular}{l|l}
\begin{subfigure}{0.6\textwidth}
\begin{tabular}{lll}
\multicolumn{3}{l}{\textbf{Generated \fstar API}}
\\
State & Edge & Generated Callback Type 
\\
$1$ & $A?\mathtt{start}$
&
{\lstinlinefst+s1 ->(n:int{0<=n<100}) ->ML unit+}
\\
$2$ & $A?\mathtt{limit}$
&
{\lstinlinefst+s2 ->(t:int{0<t}) ->ML unit+}
\\
$3$ & $C?\mathtt{guess}$
&
{\lstinlinefst+s3 ->(x:int{0<=x}) ->ML unit+}
\\
$4$ & [multiple]
&
{\lstinlinefst+(s:s4) ->ML (s4Cases s)+}
\\
$5$ & $A!\mathtt{higher}$
&
{\lstinlinefst+s5 ->ML unit+}
\\
$6$ & $A!\mathtt{lower}$
&
{\lstinlinefst+s6 ->ML unit+}
\\
$7$ & $A!\mathtt{lose}$
&
{\lstinlinefst+s7 ->ML unit+}
\\
$8$ & $A!\mathtt{win}$
&
{\lstinlinefst+s8 ->ML unit+}
\\
\end{tabular}
\caption{Generated I/O Callback Types}
\label{fig:higherlower_projection_B-types}
\end{subfigure}
&
\begin{subfigure}{0.3\textwidth}
\begin{tabular}{l}
{\begin{minipage}{\textwidth}
{\begin{lstlisting}[language=FSharp,numbers=none,basicstyle={\small\ttfamily}]
type s4Cases (s:s4) =
| s4_lower of
  unit{s.n<s.x && s.t>1}
| s4_lose of
  unit{s.n£$\neq$£s.x && s.t=1}
| s4_win of unit{s.n=s.x}
| s4_higher of
  unit{s.n>s.x && s.t>1}
\end{lstlisting}}
\end{minipage}}
\end{tabular}
\caption{Generated Data Type for the Output Choice}
\label{fig:higherlower_projection_B-internal-choice}
\end{subfigure}
\\
\end{tabular}
}

\vspace{-3mm}
\caption{Projection and \fstar{} API Generation for \textsf{\textbf{\color{Teal}B}} in
\lstinline+HigherLower+}
\vspace{-4mm}
\label{fig:higherlower_projection_B}
\end{figure}

We first look at the projection onto $\ppt B$: although it is the largest of
the three projections, it is unique among them because $\ppt B$
is involved in \emph{every} interaction of the protocol, and (consequently)
$\ppt B$ has \emph{explicit} knowledge of the value of \emph{every}
refinement variable during execution.

Formally, projection is defined as a syntactic function (explained in detail
later in
\cref{subsection:theory-projection});
it is a partial function, designed conservatively to reject protocols that are
not safely realisable in asynchronous distributed settings.
However, we show in \cref{fig:higherlower_projection_B-cfsm} the
representation of projections employed in our toolchain based on
\emph{communicating finite state machines} (CFSMs)~\cite{JACM83CFSM},
where the transitions are the localised I/O actions performed by $\ppt B$
in this protocol.
Projected CFSM actions preserve their refinements: as before, an action
refinement serves as a precondition for an output transition to be fired, and a
postcondition when an input transition is fired.
For example, $A?\mathtt{start}(n_0)\{0\leq n_0 < 100\}$ is an input of a
$\mathtt{start}$ message from $\ppt A$, with a refinement on the
\lstinline+int+ payload value.
Similarly,
$C!\mathtt{higher}\{n>x \wedge t>1\}$ expresses a protocol flow refinement on
an output of a $\mathtt{higher}$ message to $\ppt C$.
For brevity, we omit the payload data types in the CFSM edges, as this example features only
\lstinline+int+s; we omit empty payloads ``$()$'' likewise.

We show the local state refinements as annotations on the corresponding CFSM
states (shaded in grey, with an arrow to the state).

\mypara{Refined API Generation for \textsf{\textbf{\color{Teal}B}}}

CFSMs offer an intuitive understanding of the semantics of endpoint
projections.
Building on recent
work~\cite{FASE16EndpointAPI,CC18TypeProvider,POPL19Parametric}, we use our CFSM-based
representation of refined projections to \emph{generate} protocol- and
role-specific APIs for implementing each role in \fstar.
We highlight a novel and crucial development: we exploit the approach of
\emph{type} generation to produce functional-style \emph{callback}-based APIs
that \emph{internalise} all of the actual communication channels and I/O
actions.
In short, the transitions of the CFSM are rendered as a set of
\emph{transition-specific} function types to be implemented by the user --- each
of these functions take and return only the \emph{user-level data} related to
I/O actions and the running of the protocol.
The transition function of the CFSM itself is embedded into the API by the
generation, exporting a user interface to execute the protocol by calling back
the appropriate user-supplied functions according to the current CFSM state and
I/O event occurrences.

We continue with our example,
\cref{fig:higherlower_projection_B-types} lists the function types for
$\ppt B$, detailed as follows.
Note, a characteristic of MPST-based CFSMs is that each non-terminal state is
either input- or output-only.

\begin{itemize}[leftmargin=*]
  \item \textbf{State Types.}
  For each state, we generate a separate type (named by enumerating the states,
  by default).
  Each is defined as a record containing previously known
  payload values and its local recursion variables, or \keyword{unit} if
  none, for example:
\vspace{1mm}

\centerline{ \lstinlinefst+type s3 =+ $\big\{$ \lstinlinefst+n0:
  int\{0<=n0<100\}; t0: int\{0<t0\}; n: int\{0<=n<100\}; t: int\{0<t\}+ $\big\}$ }
\vspace{1mm}

  \item \textbf{Basic I/O Callbacks.}
  For each input transition we generate a function type
  \lstinlinefst+s ->+$\sigma$\lstinlinefst+ ->ML unit+, where
  \lstinline+s+ is the predecessor state type, and $\sigma$
  is the refined payload type received.
  The return type is \lstinlinefst+unit+ and the function can perform side
  effects, i.e.\ the callback is able to modify global state, interact with the
  console, etc, instead of merely pure computation.
  If an input transition is fired during execution, the generated runtime will
  invoke a user-supplied function of this type with the appropriately populated
  value of \lstinline+s+, including any payload values received in the message
  that triggered this transition.
  Note, any data or protocol refinements are embedded into the types of these
  fields.

  Similarly, for each transition of an output state with a \emph{single}
  outgoing transition, we generate a function type
  \lstinlinefst+s ->ML +$\tau$, where $\tau$ is the refined type for the
  output payload.

  \item \textbf{Internal Choices.}
  For each output state with more than 1 outgoing transition, we generate an
  additional sum type $\rho$ with the cases of the choice, e.g.\
  \cref{fig:higherlower_projection_B-internal-choice}.
  This sum type (i.e.\ \lstinline+s4Cases+) is indexed by the corresponding
  state type (i.e.\ \lstinline+s+) to make any required knowledge available for
  expressing the protocol flow refinement of each case.
  Its constructors indicate the label of the internal choice.

  We then generate a single function type for this state,
  \lstinlinefst+s ->ML +$\rho$: the user implementation selects which choice case to
  follow by returning a corresponding $\rho$ value, under the constraints of
  any refinements.
  For example, a \lstinline+s4_win+ value can only be constructed, thus this
  choice case only be selected, when \lstinlinefst+s.n=s.x+ for the given
  \lstinline+s+.
  The state machine is advanced according to the constructor of the returned
  value (corresponding to the label of the message), and the generated runtime
  sends the payload value to the intended recipient.

\end{itemize}

An asynchronous \emph{output} event, i.e.\ the trigger for the
API to call back an output function, requires the communication medium to be ready
to accept the message (e.g.\ there is enough memory in the local output
buffer).
For simplicity, in this work we consider the callbacks of an
output state to always be immediately fireable. Concretely, we delegate these concerns to the
underlying libraries and runtime system.

\mypara{Projection and API Generation for \textsf{\textbf{\color{Teal}C}}}

The projection onto $\ppt C$ raises an interesting question related to the
refinement of \emph{multiparty} protocols: how should we treat refinements on
variables that the target role does \emph{not} itself \emph{know}?
$\ppt C$ does not know the value of the secret \lstinline+n+ (otherwise
this game would be quite boring), but it does know that this information
\emph{exists} in the protocol and is subject to the specified refinement.
In standard MPST, it is essentially the main point of projection that
interactions not involving the target role can be simply (and safely)
\emph{dropped} outright; e.g.\ the communication of the \lstinline+start+
message would simply not appear in the projection of $\ppt C$.
However, naively taking the same approach in RMPST would be inadequate:
although the target role may not know some exact value, the role may still need
the associated ``latent information'' to fulfil the desired application
behaviour.

Our framework introduces a notion of \emph{erased} variables for
RMPST --- in short, our projection does drop third-party \emph{interactions}, but
retains the latent information as refinement-typed
\emph{erased} variables, as illustrated by the annotation on state $1$ in
\cref{fig:higherlower-impl-c-cfsm}.
Thanks to the SMT-based refinement type system of \fstar, the type-checker can
still take advatange of
the refined types of erased variables to \emph{prove} properties of the
endpoint implementation;
however, these variables cannot actually be used
\emph{computationally} in the implementation (since their values are not known).
Conveniently, \fstar supports erased types (described briefly in
\cref{subsection:fstar-bg}), and provides ways (i.e.\ \code{Ghost} effects) to ensure that
such variables are not used in the computation.
We demonstrate this for our example in the next subsection.
Our approach can be considered a version of \emph{irrelevant} variables
from \cite{LICS01Irrelavance, LMCS12Irrelevance} for the setting of typed,
distributed interactions.

\begin{figure}[t]
\small
\begin{subfigure}{\textwidth}
  \centering
\begin{tikzpicture}
[
  >=stealth',
  NODE/.style={draw,circle,minimum width=3mm,inner sep=0pt,fill=lightgray}
]
\node[NODE,label={[align=left,xshift=-30pt,yshift=0pt]above:{$\mathcolorbox{lightgray}{n
      \mathop{:}
      \mathtt{erased~int}
      \{0 \mathop{\leq} n \mathop{<}
      100\},}$\\[-3pt]$\mathcolorbox{lightgray}{t \mathop{:} \mathtt{erased~int} \{0 \mathop{<} t\}}$}}] (C1) {$1$};
\node[NODE,right=50mm of C1] (C2) {$2$};
\node[NODE,right=20mm of C2] (C3) {};
\node[right=12cm of C1] {};  
\draw[>=angle 60,->,densely dotted] ($(C1.north west)+(-5.5mm,2mm)$) -- (C1.north west);
\draw[->] ($(C1.west)-(5mm,0)$) -- (C1);
\draw[->] (C1) -- node[above,yshift=-2pt]{$B!\mathtt{guess}(x)\{0 \mathop{\leq} x \mathop{<} 100\}$} (C2);
\path[->] (C2) edge[bend left=20] node [below]{$B?\mathtt{lower}\{n < x \land t > 1\}$} (C1);
\path[->] (C2) edge[bend right=20] node [above]{$B?\mathtt{higher}\{n > x \land t > 1\}$} (C1);
\path[->] (C2) edge[bend left] node [above]{$B?\mathtt{win}\{n = x\}$} (C3);
\path[->] (C2) edge[bend right] node [below]{$B?\mathtt{lose}\{n \neq x \land t =
  1\}$} (C3);
\end{tikzpicture}
\caption{CFSM Representation of the Projection}
\label{fig:higherlower-impl-c-cfsm}
\end{subfigure}
\begin{subfigure}{\textwidth}
\renewcommand{\lstinlinefst}{\lstinline[language=FSharp,basicstyle=\footnotesize\ttfamily]}
{\small
\begin{tabular}{llll}
&&&
User implementation
\\
&&&
{\lstinlinefst+(* Allocate a refined int reference *)+}
\\
State & Edge & Generated type &
{\lstinlinefst+let next: ref (x:int{0<=x<100}) = alloc 50+}
\\
\hline
$1$ & $B!\mathtt{guess}$
&
{\lstinlinefst+s1 ->ML (x:int{0<=x<100})+}
&
{\lstinlinefst+fun _ ->!next (*Deref next*)+}
\\
$2$ & $B?\mathtt{higher}$
&
{\lstinlinefst+s2 ->unit{n>x && t>1} ->ML unit+}
&
\red{{\lstinlinefst*fun s ->next := s.x + 1*}}
\\
 & $B?\mathtt{lower}$
&
{\lstinlinefst+s2 ->unit{n<x && t>1} ->ML unit+}
&
\red{{\lstinlinefst*fun s ->next := s.x - 1*}}
\\
 & $B?\mathtt{win}$
&
{\lstinlinefst+s2 ->unit{n=x} ->ML unit+}
&
{\lstinlinefst!fun _ ->()!}
\\
 & $B?\mathtt{lose}$
&
{\lstinlinefst+s2 ->unit{n<>x && t=1} ->ML unit+}
&
{\lstinlinefst!fun _ ->()!}
\end{tabular}
}
\caption{Generated I/O Callback Types}
\label{fig:higherlower-impl-c-types}
\end{subfigure}
\vspace{-3mm}
\caption{Projection and \fstar{} API Generation for \textsf{\textbf{\color{Teal}C}} in
\lstinline+HigherLower+
}
\label{fig:higherlower-impl-c}
\vspace{-5mm}
\end{figure}

\subsection{\fstar Implementation -- Protocol Validation and Verification by Refinement Types}
\label{subsection:impl-fstar-impl}

Finally, the generated APIs---embodying the refined projections---are used
to implement the endpoint processes.
As mentioned, the user implements the program logic as callback functions of
the generated (refinement) types, supplied to the entry point along with code for
establishing the communication channels between the session peers.
Assuming a record \lstinline+callbacks+ containing the required functions
(static typing ensures all are covered), \cref{fig:c-impl-main} bootstraps a
$\ppt C$ endpoint.

\begin{figure}[h]
  \centering
\begin{subfigure}{0.48\textwidth}
{\begin{lstlisting}[language=FSharp,numbers=none,basicstyle={\small\ttfamily}]
let main () =
  (* connect to B via TCP *)
  let server_B = connect ip_B port_B in
  (* Setup connection from TCP *)
  let conn = mk_conn server_B in
  (* Invoke the Entry Point `run` *)
  let () = run callbacks conn in
  (* Close TCP connection *)
  close server_B
\end{lstlisting}}
\vspace{-2mm}
\caption{Running the Endpoint
  \textsf{\textbf{\color{Teal}C}}}\label{fig:c-impl-main}
\end{subfigure}
\begin{subfigure}{0.48\textwidth}
{\begin{lstlisting}[language=FSharp,numbers=none,basicstyle={\small\ttfamily}]
(* Signature (s:s4) -> ML (s4Cases s) *)
fun (s:s4) ->
  (* Win if guessed number is correct *)
  if s.x=s.n then s4_win ()
  (* Lose if running out of attempts *)
  else if s.t=1 then s4_lose ()
  (* Otherwise give hints accordingly *)
  else if s.n>s.x then s4_higher ()
  else s4_lower ()
\end{lstlisting}}
\vspace{-2mm}
\caption{Implementing the Internal Choice for
  \textsf{\textbf{\color{Teal}B}}}\label{fig:c-impl-internal-choice}
\end{subfigure}
\vspace{-3mm}
\caption{Selected Snippets of Endpoint Implementation}
\vspace{-5mm}
\end{figure}

The API takes care of endpoint execution by monitoring the channels, and
calling the appropriate callback based on the current protocol state and I/O
event occurrences.
For example, a minimal, well-typed implementation of $\ppt B$ could
comprise the internal choice callback above (\cref{fig:c-impl-internal-choice})
(implementing the type in \cref{fig:higherlower_projection_B-internal-choice}), cf.\ state $4$,
and an empty function for all others
(i.e.\ \lstinlinefst+fun _ ->()+).
We can highlight how protocol violations are ruled out by static refinement
typing, which is ultimately the practical purpose of RMPST.
In the above callback code, changing, say, the condition for the
\lstinline+lose+ case to \lstinlinefst+s.t=0+ would directly violate the
refinement on the \lstinline+s4_lose+ constructor, cf.\
\cref{fig:higherlower_projection_B-internal-choice}.
Omitting the \lstinline+lose+ case altogether would break both the
\lstinline+lower+ and \lstinline+higher+ cases, as the type checker would not
be able to prove \lstinlinefst+s.t>1+ as required by the subsequent
constructors.

Lastly, \cref{fig:higherlower-impl-c-types} implements $\ppt C$ to
guess the secret by a simple search, given we know its value is bounded within
the specified interval.
We draw attention to the input callback for \lstinline+higher+, where we adjust
the \lstinline+next+ value.
Given that the value being assigned is one more than the existing value, it
might have been the case that the new value falls out of the range (in the case
where \lstinline+next+ is 99), hence violating the prescribed type.
However, despite that the value of \lstinline+n+ is unknown, we have known from
the refinement attached to the edge that \lstinlinefst+n>x+ holds, hence it must
have been the case that our last guess \lstinline+x+ is strictly less than the
secret \lstinline+n+, which rules out the possibility that
\lstinline+x+ can be 99 (the maximal value of \lstinline+n+).
Had the refinement and the erased variable not been present, the type-checker
would not be able to accept this implementation, and it demonstrates that our
encoding allows such reasoning with latent information from the protocol.

Moreover, the type and effect system of \fstar prevents the erased variables
from being used in the callbacks.
On one hand, \lstinlinefst{int} and \lstinlinefst{erased int} types are not
compatible, because they are not the same type.
This prevents an irrelevant variable from being used in place of a concrete
variable.
On the other hand, the function \lstinlinefst{reveal} converts a value of
\lstinlinefst{erased 'a} to a value of \lstinlinefst{'a} with \lstinlinefst{Ghost}
effect.
A function with \lstinlinefst{Ghost} effect \emph{cannot} be mixed with a
function with \lstinlinefst{ML} effect (as in the case of our callbacks), so
irrelevant variables cannot be used in the implementation via the \lstinlinefst{reveal}
function.

Interested readers are invited to try the running example out with our
accompanying artifact.
We propose a few modifications on the implementation code and the protocol, and
invite the readers to observe errors when implementations no longer correctly
conforms to the prescribed protocol.

\subsection{Executing the Communicating Finite State Machine (Generated Code)}
\label{subsection:impl-run-fsm}
As mentioned earlier, our API design sepearates the concern of program logic
(with callbacks) and communication (with connections).
A crucial piece of the generated code involves \emph{threading} the two parts
together --- the execution function performs the
communications actions and invokes the appropriate callbacks for handling.
In this way, we do \emph{not} expose explicit communication channels, so
linearity can be achieved with ease by construction in our generated code.

The entry point function, named \code{run}, takes callbacks and connections as
arguments, and executes the CFSM for the specified endpoint.
The signature uses the permissive \keyword{ML} effect, since communicating with
the external world performs side effects.
We traverse the states (the set of states is denoted $\Q$) in the CFSM and
generate appropriate code depending on the nature of the state and its outgoing
transitions.

Internally, we define mutually recursive functions for each state $q \in \Q$,
named \code{run$_q$},
taking the state record $\enc{q}$ as argument ($\enc{q}$ stands for the state record
for a given state $q$),
which performs the required actions at state $q$.
The run state function for a state $q$ either (1) invokes callbacks and communication primitives,
then calls the  run state function for the successor state $q'$,
or (2) returns directly for termination if $q$ is a
terminal state (without outgoing transitions).
The main entry point invokes the run function for the initial state $q_0$,
starting the finite state machine.

The internal run state functions are not exposed to the developer, hence
it is not possible to tamper with the internal state with usual means of
programming.
This allows us to guarantee linearity of communication channels by
construction.
In the following text, we outline how to run each state, depending on whether
the state is a sending state or a receiving state.
Note that CFSMs constructed from local types do not have mixed states
\cite[Prop. 3.1]{ICALP13CFSM}

\begin{figure}[h]
  \begin{subfigure}{0.48\textwidth}
    \begin{lstlisting}[language=FSharp,numbers=none]
let rec run_£$q$£ (st: state£$q$£) =
  let choice = callbacks.state£$q$£_send st
  in match choice with
  | Choice£$q\dlbl{l_i}$£ payload -> £$\tikzmark{curly-brace-1-start}$£
    comm.send_string £$\ppt q$£ "£$\dlbl{l_i}$£";
    comm.send_£$\dte S$ \ppt q£ payload;
    let st = { £$\cdots$£; £$\dexp{x_i}$£=payload } in £$\tikzmark{curly-brace-1-longest}$£
    run_£$q'$£ st £$\tikzmark{curly-brace-1-end}$£
    \end{lstlisting}
    \AddNote{curly-brace-1-start}{curly-brace-1-end}{curly-brace-1-longest}{Repeat
      for $i \in I$}
    \vspace{-8mm}
    \caption{Template for Sending State $q$}
    \label{fig:runq-sending}
  \end{subfigure}
  \hfill
  \begin{subfigure}{0.48\textwidth}
    \begin{lstlisting}[language=FSharp, numbers=none]
let rec run_£$q$£ (st: state£$q$£) =
  let label = comm.recv_string £$\ppt p$£ () in
  match label with
  £$\tikzmark{curly-brace-2-start}$£| "£$\dlbl{l_i}$£" ->
    let payload = comm.recv_£$\dte S$£ £$\ppt p$£ () in
    callbacks.state£$q$£_receive_£$\dlbl{l_i}$£ st payload;
    let st = { £$\cdots$£; £$\dexp{x_i}$£=payload } in
  £$\tikzmark{curly-brace-2-end}$£  run_£$q'$£ st
    \end{lstlisting}
    \AddNoteLeft{curly-brace-2-start}{curly-brace-2-end}{curly-brace-2-end}{}
    \vspace{-8mm}
    \caption{Template for Receiving State $q$}
    \label{fig:runq-receiving}
  \end{subfigure}
  \vspace{-3mm}
  \caption{Template for \code{run}$_q$}
  \vspace{-3mm}
\end{figure}
\mypara{Running the CFSM at a Sending State}
For a sending state $q \in \Q$, the developer makes an internal choice on how
the protocol proceeds, among the possible outgoing transitions.
This is done by invoking the sending callback $\code{state}q\code{\_send}$ with
the current state record, to obtain a choice with the associated payload.
We pattern match on the constructor of the label $\dlbl{l_i}$ of the message, and
find the corresponding successor state $q'$.

The label $\dlbl{l_i}$ is encoded as a \tstr\ and sent via the sending
primitive to $\ppt q$.
It is followed by the payload specified in the return value of the callback,
via corresponding sending primitive according to the base type with refinement
erased.

We construct a state record of $\enc{q'}$ from the existing record $\enc{q}$,
adding the new field $\dexp{x_i}$ in the action using the callback return
value.
In the case of recursive protocols, we also update the recursion variable
according to the definition in the protocol when constructing $\enc{q'}$.
Finally, we call the run state function $\code{run}_{q'}$ to continue the
CFSM, effectively making the transition to state $q'$.

Following the procedure, $\code{run}_q$ is generated as shown in \cref{fig:runq-sending}.

\mypara{Running the CFSM at a Receiving State}
For a receiving state $q \in \Q$, how the protocol proceeds is determined by
an external choice, among the possible outgoint actions.
To know what choice is made by the other party, we first receive a string
and decode it into a label $\dlbl{l}$, via the receiving primitive for string.

Subsequently, according to the label $\dlbl{l}$, we can look up the label in
the possible transitions, and find the state successor $q'$.
By invoking the appropriate receiving primitive, we obtain the payload value.
We note that the receiving primitive has a return type without refinements.
In order to re-attach the refinements, we use the \fstar builtin \keyword{assume} to
reinstate the refinements according to the protocol before using the value.

According the label $\dlbl{l}$ received, we can call the corresponding
receiving callback with the received value.
This allows the developer to process the received value and perform any
relevant program logic.
This is followed by the same procedure for constructing the state record for
the next state $q'$ and invoking the run function for $q'$.

Following the procedure, $\code{run}_q$ is generated as shown in \cref{fig:runq-receiving}.

\subsection{Summary} 

We demonstrated with our running example, \lstinlinefst{HigherLower}, how to
implement a refined multiparty protocol with our toolchain \Ourtool.

Exploiting the powerful type system of \fstar, our approach has several key
benefits:
First, it guarantees \emph{fully static} session type safety in a
lightweight, practical manner --- the callback-style API is portable to any
statically typed language.
Existing work based on code generation has considered only hybrid approaches
that supplement static typing with dynamically checked linearity of explicit
communication channel usages.
Moreover, the separation of program logic and communication leads to a modular
implementation of protocols.

Second, it is well suited to functional languages like \fstar; in particular,
the \emph{data}-oriented nature of the user interface allows the refinements in
RMPST to be directly mapped to data refinements in \fstar, allowing the
refinements constraints to be discharged at the user implementation level by the
\fstar compiler --- again, fully statically.

Furthermore, our endpoint implementation inherits core communication safety
properties such
as freedom from deadlock or communication mismatches, based on the original
MPST theory.
We use the \fstar type-checker to validate statically that an endpoint
implementation is correctly typed with respect to the prescribed type obtained
via projection of the global protocol.
Therefore, the implementation benefits from additional guarantees from the
refinement types.


%% file: fstar-background.tex
\subsection{Targeting \fstar and Implementing Endpoints}
\label{subsection:fstar-bg}

\fstar \cite{POPL16FStar} is a verification-oriented programming language, with
a rich set of features.
Our method of API generation and example programs utilise the following \fstar
features:%
\footnote{A comprehensive \fstar{} tutorial is available at
\url{https://www.fstar-lang.org/tutorial/}.}

\begin{itemize}[leftmargin=*]
  \item \textbf{Refinement Types.}
  A \emph{refinement type} has the form
  \lstinline+x:t{E}+, where \lstinline+t+ is the base
  type, \lstinline+x+ is a variable that stands for values of this type, and
  \lstinline+E+ is a pure\footnotemark boolean expression for \emph{refinement},
  possibly containing \lstinline+x+.
  \footnotetext{Pure in this context means pure terminating computation, i.e.\
  no side-effects including global state modifcations, I/O actions or infinite
  recursions, etc.}
  In short, the values of this refinement type are the \emph{subset} of values
  of \lstinline+t+ that make \lstinline+E+ evaluate to \lstinline+true+,
  e.g.\ natural numbers are defined as
  \lstinline[language=FSharp,basicstyle=\small\ttfamily]+x:int{x+{\small $\geq$}\lstinline+0}+.
  We use this feature to express data and control flow constraints in protocols.

  In \fstar, type-checking refinement types are done with the assistance of the
  Z3 SMT solver~\cite{TACAS08Z3}.
  Refinements are encoded into SMT formulae and the solver decides the
  satisfiability of SMT formulae during type-checking.
  This feature enables automation in reasoning and saves the need for manual
  proofs in many scenarios.

  \item \textbf{Indexed Types.}
  \emph{Types} can take pure expressions as arguments.
  For example, a declaration \newline
  \lstinline[language=FSharp,basicstyle=\small\ttfamily]+type t (i:t') = ...+ prescribes the
  \emph{family} of types given by applying the type constructor \lstinline+t+
  to \emph{values} of type \lstinline+t'+.
  We use this feature to generate type definitions for payload items in an
  internal choice, where the refinements in payload types refer to a state type.

  \item \textbf{Dependent Functions with Effects.}
  A (dependent) function in \fstar has a type of the form \linebreak
  \lstinline+(x:t+${}_1$\lstinline+)+ $\to$\lstinline+ E t+${}_2$, where
  \lstinline+t+${}_1$ is the argument type, \lstinline+E+ describes
  the \emph{effect} of the function, and \lstinline+t+${}_2$ is the result
  type, which may also refer to the argument \lstinline+x+.

  The default effect is \lstinline+Tot+, for pure total expressions
  (i.e.\ terminating and side-effect free).
  At the other end of the spectrum is the arbitrary effect \lstinline+ML+
  (correspondent to all possible side effects in an ML language),
  which permits state mutation, non-terminating recursion, I/O,
  exceptions, etc.

  \item \textbf{The \code{Ghost} Effect and the \code{erased} Type.}
  A type can be marked \keyword{\small erased} in \fstar, so that values of such types
  are not available for computation (after extracting into target
  language), but only for proof purposes (during type-checking).
  The type constructor is accompanied with the \code{Ghost} effect to mark
  computationally irrelevant code, where
  the type system prevents the use of erased values in
  computationally relevant code, so that the values can really be safely erased.
  In the following snippet, we quickly demonstrate this feature: \code{GTot}
  stands for \code{Ghost} and total, and cannot be mixed with the default
  pure effect (the function \code{not\_allowed} does not type-check).
  We use the \keyword{erased} type to mark variables known to the endpoint via
  the protocol specification, whose values are not known due to not being a
  party of the message interaction. For example, in \cref{fig:guess}, the endpoint \code{C} does not
  know the value of \code{n0}, but knowns its type from the protocol.

\begin{center}
\begin{minipage}{0.45\textwidth}
\begin{lstlisting}[language=FSharp]
type t = { x1: int;
           x2: erased int; }
(* Definition in standard library *)
val reveal: erased a -> GTot a
\end{lstlisting}
\end{minipage}
\begin{minipage}{0.5\textwidth}
\begin{lstlisting}[language=FSharp]
(* The following access is not allowed *)
let not_allowed (o: t) = reveal o.x2
(* Accessing at type level is allowed *)
val allowed: (o: t{reveal o.x2 >= 0}) ->int
\end{lstlisting}
\end{minipage}
\end{center}

\end{itemize}

Our generated code consists of multiple type definitions and an entry point
function (as shown in \cref{fig:workflow-detail}, \fstar API), including:
\begin{description}[leftmargin=*]
  \item[State Types:] Allowing developers to access variables known at a given
    CFSM state.
  \item[Callbacks:] A record of functions corresponding to CFSM transitions,
    used to implement program logic of the local endpoint.
  \item[Connections:] A record of functions for sending and receiving values to
    and from other roles in the global protocol, used to handle the
    communication aspects of the local endpoint.
  \item[Entry Point:] A function taking callbacks and
    connections to run the local endpoint.
\end{description}

To implement an endpoint, the developer needs to provide implementations for the
generated callback and connection types, using appropriate functions to handle the
program logic and communications.
The \fstar compiler checks whether the implemented functions type-check against
the prescribed types.
If the endpoint implementation succeeds the type-checking, the developer may
choose to extract to a target language (e.g.\ OCaml, C) for execution.


%% file: theory.tex
\section{A theory of Refined Multiparty Session Types (RMPST)}
\label{section:theory}

In this section, we introduce \emph{refined multiparty session types}
(\emph{RMPST} for short).
We give the syntax of types in \cref{subsection:theory-syntax}, extending
original multiparty session types (MPST) with \emph{refinement types}.
We describe the refinement typing system that we use to type expressions in
RMPST in \cref{subsection:typing-exp}.

We follow the standard MPST methodology.
\emph{Global session types} describe communication structures of many
\emph{participants} (also known as \emph{roles}).
\emph{Local session types}, describing communication structures of a single
participant, can be obtained via \emph{projection} (explain in
\cref{subsection:theory-projection}).
Endpoint processes implement local types obtained from projection.
We give semantics of global types and local types in
\cref{subsection:theory-semantics}, and show the equivalence of semantics
with respect to projection.
As a consequence, we can compose all endpoint processes implementing local
types for roles in a global type, obtained via projection, to implement
the global type correctly.

\subsection{Syntax of Types}
\label{subsection:theory-syntax}

We define the syntax of refined multiparty session types (refined MPST) in
\cref{fig:mpst-syntax}.
We use different colours for different syntactical categories to help
disambiguation, but the syntax can be understood without colours.
We use \dgt{pink}~for global types, \dtp{dark blue}~for local types,
\dexp{blue}~for expressions, \dte{purple}~for base types, \dlbl{indigo}~for
labels, and \ppt{Teal}~with bold fonts for participants.

\begin{figure}[h]
  \vspace{-1mm}
  \begin{tabular}{ll}
    $
    \arraycolsep=2pt
    \begin{array}{rcll}
      \dte{S} & \bnfas & \tint \bnfalt \tbool \bnfalt \dots & \text{\footnotesize Base Types} \\
      \dte{T} & \bnfas & \dexp x:\dte{S}\esetof{E} & \text{\footnotesize Refinement Types} \\
      \dexp{E} & \bnfas & \dexp x \bnfalt \dexp {\underline n} \bnfalt op_1~\dexp{E} \bnfalt \dexp{E} ~op_2~ \dexp{E} \dots & \text{\footnotesize Expressions} \\
      \dgt{G} & \bnfas & & \text{\footnotesize Global Types} \\
      & \bnfalt & \gtbran{p}{q}{\dlbl{l_i}(\vb{x_i}{T_i}). G_i}_{i \in I} &
      \text{\footnotesize Message} \\
       & \bnfalt & \gtrecur{t}{\vb{x}{T}}{\dexp{x := E}}{G} &
       \text{\footnotesize Recursion} \\
       & \bnfalt & \gtvar{t}{\dexp{x := E}} ~~\bnfalt~~ \gtend & \text{\footnotesize Type Var.,
       End} \\
    \end{array}$
  &
    $
    \begin{array}{rcll}
      \dtp{L} &
      \bnfas & &  \text{\footnotesize Local Types} \\
      & \bnfalt &
        \toffer{p}{\dlbl{l_i}(\vb{x_i}{T_i}). L_i}_{i \in I} &
        \text{\footnotesize Receiving}\\
      & \bnfalt &
        \ttake{p}{\dlbl{l_i}(\vb{x_i}{T_i}). L_i}_{i \in I} &
        \text{\footnotesize Sending}\\
      & \bnfalt & \tphi{l}{x}{T}{L} & \text{\footnotesize Silent Prefix} \\
      & \bnfalt &
        \trecur{t}{}{\vb{x}{T}}{\dexp{x := E}}{L} & \text{\footnotesize
          Recursion} \\
      & \bnfalt & \tvar{t}{\dexp{x := E}} ~~\bnfalt~~ \tend & \text{\footnotesize Type Var.,
        End} \\
    \end{array}$
  \end{tabular}
\vspace{-1mm}
\caption{Syntax of Refined Multiparty Session Types}
\vspace{-3mm}
\label{fig:mpst-syntax}
\end{figure}

\mypara{Value Types and Expressions}
We use $\dte{S}$ for base types of values, ranging over integers, booleans,
etc.
Values of the base types must be able to be communicated.

The base type $\dte{S}$ can be \emph{refined} by a boolean expression,
acting as a predicate on the members of the base type.
A \emph{refinement type} is of the form $(\vb{x}{S} \esetof{E})$.
A value $\dexp{x}$ of the type has base type $\dte{S}$, and is refined by a
boolean expression $\dexp E$.
The boolean expression $\dexp E$ acts as a predicate on the members $\dexp x$
(possibly involving the variable $\dexp x$).
For example, we can express natural numbers as $(\vb{x}{\tint} \esetof{x \geq
0}$).
We use $\fv{\cdot}$ to denote the free variables in refinement types, expressions, etc.
We consider variable $\dexp x$ be bound in the refinement expression $\dexp
E$, i.e.
$\fv{\vb{x}{S}\esetof{E}} = \fv{\dexp E} \setminus \esetof{x}$.

Where there is no ambiguity, we use the base type $\dte S$ directly as an
abbreviation of a refinement type $(\vb{x}{S}\esetof{\etrue})$, where $\dexp x$
is a fresh variable, and $\etrue$ acts as a predicate that accepts all values.

\mypara{Global Session Types}
\emph{Global session types} (\emph{global types} or \emph{protocols} for short) range over
$\dgt{G, G', G_i, \dots}$
Global types give an overview of the overall communication structure.
We extend the standard global types \cite{ICALP13CFSM} with
refinement types and variable bindings in message prefixes. Extensions to the syntax
are \shaded{\text{shaded}} in the following explanations.

$\gtbran{p}{q}{\dlbl{l_i}(\shaded{\dexp{x_i}:}\dte{T_i}). G_i}_{i \in I}$ is a
message from $\ppt p$ to $\ppt q$, which branches into one or more
continuations with label $\dlbl{l_i}$, carrying a payload variable
$\dexp{x_i}$ with type $\dte{T_i}$.
We omit the curly braces when there is only one branch, like
$\gtbransingle{p}{q}{\dlbl{l}(\vb{x}{T})}$.
We highlight the difference from the standard syntax, i.e.\ the variable binding.
The payload variable $\dexp{x_i}$ occurs bound in the continuation global
type $\dgt{G_i}$, for all $i \in I$.
We sometimes omit the variable if it is not used in the continuations.
The free variables are defined as:
\vspace{-1mm}
$$\fv{\gtbran{p}{q}{\dlbl{l_i}(\vb{x_i}{T_i}). G_i}_{i \in I}}
=
\bigcup_{i \in I}{\fv{\dte{T_i}}}
  \cup
\bigcup_{i \in I}{(\fv{\dgt{G_i}} \setminus \esetof{x_i})}
\vspace{-1mm}
$$
We require that the index set $I$ is not empty, and all labels $\dlbl{l_i}$
are distinct.
To prevent duplication, we write $\dlbl{l}(\vb{x}{S}\esetof{E})$
instead of $\dlbl{l}(\vb{x}{(\vb{x}{S}\esetof{E})})$ (the first $\dexp x$ occurs
as a variable binding in the message, the second $\dexp x$ occurs as a variable
representing member values in the refinement types).

We extend the construct of recursive protocols to include a variable carrying
a value in the inner protocol.
In this way, we enhance the expressiveness of the global types by allowing
a recursion variable to be maintained across iterations of global protocols.

The recursive global type
$\gtrecur{t}{\shaded{\vb{x}{T}}}{\shaded{\dexp{x := E}}}{G}$
specifies a variable $\dexp{x}$ carrying type $\dte{T}$ in the
recursive type, initialised with expression $\dexp E$.
The type variable $\gtvar{t}{\shaded{\dexp{x := E}}}$ is annotated with an
assignment of expression $\dexp E$ to variable $\dexp x$.
The assignment updates the variable $\dexp{x}$ in the current recursive
protocol to expression $\dexp{E}$.
The free variables in recursive type is defined as
\vspace{-1mm}
$$\fv{\gtrecur{t}{{\vb{x}{T}}}{{\dexp{x := E}}}{G}}
=
\fv{\dte T} \cup \fv{\dexp E} \cup (\fv{\dgt G} \setminus \esetof{x})$$
\vspace{-5mm}

We require that recursive types are contractive \cite[\S 21]{PierceTAPL}, so that
recursive protocols have at least a message prefix, and protocols such as
$\gtrecur{t}{\vb{x}{T}}{\dexp{x:=E_1}}{\gtvar{t}{\dexp{x:=E_2}}}$ are not allowed.
We also require recursive types to be closed with respect to type variables,
e.g.\ protocols such as $\gtvar{t}{\dexp{x:=E}}$ alone are not allowed.

We write $\dgt G\subst{\gtrecursimpl{t}{\vb{x}{T}}{G}}{\dgt{\mathbf{t}}}$ to
substitute all occurrences of type variables with expressions
$\gtvar{t}{\dexp{x:= E}}$ into $\gtrecur{t}{\vb{x}{T}}{\dexp{x:=E}}{G}$. We
write $\ppt r \in \dgt G$ to say $\ppt r$ is a participating role in
the global type $\dgt G$.

\begin{example}[Global Types]
  We give the following examples of global types.
  \begin{enumerate}
    \item
      $
        \dgt{G_1} =
          \gtbransingle{A}{B}{\dlbl{Fst}(\vb{x}{\tint})}.
          \gtbransingle{B}{C}{\dlbl{Snd}(\vb{y}{\tint\esetof{x = y}})}.
          \gtbransingle{C}{D}{\dlbl{Trd}(\vb{z}{\tint\esetof{x = z}})}.
          \gtend.
      $

      \vspace{2mm}
      $\dgt{G_1}$ describes a protocol where $\ppt A$ sends an $\tint$ to $\ppt B$,
      and $\ppt B$ relays the same $\tint$ to $\ppt C$, similar for $\ppt C$ to
      $\ppt D$. Note that we can write $\dexp{x=z}$ in the refinement of $\dexp
      z$, whilst $\dexp x$ is not known to $\ppt C$.
    \vspace{2mm}
    \item
      $
        \dgt{G_2} =
          \gtbransingle{A}{B}{\dlbl{Number}(\vb{x}{\tint})}.
          \gtbran{B}{C}{
            \begin{array}{@{}l@{}}
              \dlbl{Positive}({\tunit}\esetof{x > 0}).\gtend\\
              \dlbl{Zero}({\tunit}\esetof{x = 0}).\gtend\\
              \dlbl{Negative}({\tunit}\esetof{x < 0}).\gtend
            \end{array}
          }
      $

      \vspace{2mm}
      $\dgt{G_2}$ describes a protocol where $\ppt A$ sends an $\tint$ to $\ppt B$,
      and $\ppt B$ tells $\ppt C$ whether the $\tint$ is positive, zero, or negative.
      We omit the variable here since it is not used later in the continuation.
    \vspace{2mm}
    \item
      $
      \arraycolsep=1pt
      \begin{array}[t]{@{}ll@{}}
        \dgt{G_3} = &
         \gtrecur{t}
         {\vb{try}{\tint\esetof{try \geq 0 \land try \leq 3}}}
         {\dexp{try := 0}}
         {}\\
           &\begin{array}{@{}l@{}}
              \gtbransingle{A}{B}{\dlbl{Password}(\vb{pwd}{\tstr})}.\\
              \gtbran{B}{A}{
                \begin{array}{@{}l@{}}
                  \dlbl{Correct}({\tunit}).\gtend\\
                  \dlbl{Retry}({\tunit}\esetof{try < 3}).\gtvar{t}{\dexp{try:=try+1}}\\
                  \dlbl{Denied}({\tunit}\esetof{try = 3}).\gtend
                \end{array}
              }
           \end{array}
        \end{array}
      $

      \vspace{2mm}
      $\dgt{G_3}$ describes a protocol where $\ppt A$ authenticates with $\ppt B$ with maximum 3 tries.
  \end{enumerate}
  \label{example:gty}
\end{example}

\mypara{Local Session Types}
\emph{Local session types} (\emph{local types} for short) range over
$\dtp{L, L', L_i, \dots}$
Local types give a view of the communication structure of an endpoint,
usually obtained from a global type.
In addition to standard syntax, the recursive types are
similarly extended as those of global types.

Suppose the current role is $\ppt q$,
the local type $\ttake{p}{\dlbl{l_i}(\vb{x_i}{T_i}). L_i}_{i \in I}$ describes
that the role $\ppt q$ sends a message to the partner role $\ppt p$ with label
$\dlbl{l_i}$ (where $i$ is selected from an index set $I$), carrying payload
variable $\dexp{x_i}$ with type $\dte{T_i}$, and continues with $\dtp{L_i}$.
It is also said that the role $\ppt q$ takes an \emph{internal choice}.
Similarly, the local type $\toffer{p}{\dlbl{l_i}(\vb{x_i}{T_i}). L_i}_{i \in
  I}$ describes that the role $\ppt q$ receives a message from the partner role $\ppt p$.
In this case, it is also said that the role $\ppt q$ offers an \emph{external choice}.
We omit curly braces when there is only a single branch (as is done for global
messages).

We add a new syntax construct of $\tphi{l}{x}{T}{L}$ for \emph{silent local types}.
We motivate this introduction of the new prefix to represent knowledge
obtained from the global protocol, but not in the form of a message.
Silent local types are useful to model variables obtained with
irrelevant quantification \cite{LICS01Irrelavance, LMCS12Irrelevance}.
These variables can be used in the construction of a type, but cannot be used
in that of an expression, as we explain later in \cref{subsection:typing-exp}.
We show an example of a silent local type later in \cref{example:proj}, after we
define \emph{endpoint projection}, the process of obtaining local types from a
global type.

\subsection{Expressions and Typing Expressions}
\label{subsection:typing-exp}

We use $\dexp{E, E', E_i}$ to range over expressions.
Expressions consist of variables $\dexp x$, constants (e.g.\ integer literals
$\dexp{\underline{n}}$), and unary and binary operations.
We use an SMT assisted refinement type system for typing expressions,
in the style of \cite{PLDI08LiquidTypes}.
The simple syntax of expressions allows all expressions to be encoded into
SMT logic, for deciding a semantic subtyping relation of refinement types
\cite{JFP12SemanticSubtyping}.

\input{fig/rules/typing-exp.tex}

\mypara{Typing Contexts}
We define two categories of typing contexts, for use in handling global types
and local types respectively.
$$
    \Gamma \bnfas \emptyctx \bnfalt \ctxc{\Gamma}{x^{\ppt{$\mathbb{P}$}}}{T} \hspace{8mm}
    \Sigma \bnfas \emptyctx \bnfalt \ctxc{\Sigma}{x^\theta}{T} \hspace{8mm}
    \theta \bnfas 0 \bnfalt \omega
$$
We annotate global and local typing contexts differently.
For global contexts $\Gamma$, variables carry the annotation of a set of roles
$\bigP$, to denote the set of roles that have the knowledge of its value.

For local contexts $\Sigma$, variables carry the annotation of their
multiplicity $\theta$.
A variable with multiplicity $0$ is an irrelevantly quantified variable
(irrelevant variable for short), which
cannot appear in the expression when typing (also denoted as $\dexp x \div
\dte T$ in the literature \cite{LICS01Irrelavance,LMCS12Irrelevance}).
Such a variable can only appear in an expression used as a predicate, when
defining a refinement type.
A variable with multiplicity $\omega$ is a variable without restriction.
We often omit the multiplicity $\omega$.

\mypara{Well-formedness}
Since a refinement type can contain free variables, it is necessary to define
well-formedness judgements on refinement types, and henceforth on typing contexts.

We define $\Sigma^+$ to be the local typing context where all irrelevant
variables $\dexp{x^0}$ become unrestricted $\dexp{x^\omega}$, i.e.\
$(\emptyctx)^+ = \emptyctx; (\Sigma, \vb{x^\theta}{T})^+ = \Sigma^+, \vb{x^\omega}{T}$.

We show the well-formedness judgement of a refinement type \ruleWfRty~in
\cref{fig:typing-expression}.
For a refinement type $(\vb{x}{S}\esetof{E})$ to be a well-formed type, the
expression $\dexp E$ must have a boolean type under the context $\Sigma^+$,
extended with variable $\dexp x$ (representing the members of the type) with
type $\dte S$.
The typing context $\Sigma^+$ promotes the irrelevant quantified variables
into unrestricted variables, so they can be used in the expression $\dexp E$
inside the refinement type.

The well-formedness of a typing context is defined inductively, requiring all
refinement types in the context to be well-formed.
We omit the judgements for brevity.

\mypara{Typing Expressions}
We type expressions in local contexts, forming judgements of form \linebreak
\hfill \hfill
$\Sigma \vdash \dexp E: \dte T$, and show key typing rules in
\cref{fig:typing-expression}.
We modify the typing rules in a standard refinement \linebreak type system
\cite{PLDI08LiquidTypes, ICFP14LiquidHaskell, POPL18Refinement},
adding distinction between irrelevant and unrestricted variables.

\ruleTEConst\ gives constant values in the expression a refinement type
that only contains the constant value.
Similarly, \ruleTEPlus\ gives typing derivations for the plus operator,
with a corresponding refinement type encoding the addition.

We draw attention to the handling of variables (\ruleTEVar).
An irrelevant variable in the typing context cannot appear
in an expression, i.e.\ there is \emph{no} derivation for
$\Sigma_1, \vb{x^0}{T}, \Sigma_2 \vdash \vb{x}{T}$.
These variables can only be used in an refinement type (see \ruleWfRty).

The key feature of the refinement type system is the semantic subtyping
relation decided by SMT \cite{JFP12SemanticSubtyping}, we describe the feature in \ruleTESub.
We use $\enc{\dexp E}$ to denote the encoding of expresion $\dexp E$ into the
SMT logic.
We encode a type binding $\vb{x^\theta}{(\vb{v}{S}\esetof{E})}$ in a typing
context by encoding the term $\dexp{E\subst{x}{v}}$, and define the encoding
of a typing context $\enc{\Sigma}$ inductively.

We define the extension of typing contexts
($\ctxext{\Gamma}{x^{\bigP}}{T}$; $\ctxext{\Sigma}{x^\theta}{T}$)
in \cref{fig:ctx-ext},
used in definitions of semantics.
We say a global type $\dgt G$ (resp.\ a local type $\dtp L$) is closed under
a global context $\Gamma$ (resp.\ a local context $\Sigma$), if all free variables
in the type are in the domain of the context.
\begin{figure}
  \scalebox{0.9}{$
    \ctxext{\Gamma}{x^{\bigP}}{T} = \begin{cases}
      \Gamma, \vb{x^{\bigP}}{T} & \text{if}~\dexp{x} \notin \Gamma \\
      \Gamma_1, \vb{x^{\bigP}}{T}, \Gamma_2 & \text{if}~\Gamma = \Gamma_1, \vb{x^{\varnothing}}{T}, \Gamma_2\\
      \Gamma_1, \vb{x^{\bigP}}{T}, \Gamma_2 & \text{if}~\Gamma = \Gamma_1, \vb{x^{\bigP}}{T}, \Gamma_2\\
      \text{undefined} & \text{otherwise}
    \end{cases}
    \hspace{5mm}
    \ctxext{\Sigma}{x^\theta}{T} = \begin{cases}
      \Sigma, \vb{x^\theta}{T} & \text{if}~\dexp{x} \notin \Sigma \\
      \Sigma_1, \vb{x^\theta}{T}, \Sigma_2 & \text{if}~\Sigma = \Sigma_1, \vb{x^0}{T}, \Sigma_2\\
      \Sigma_1, \vb{x^\omega}{T}, \Sigma_2 & \text{if}~\Sigma = \Sigma_1, \vb{x^\omega}{T}, \Sigma_2\\
      \text{undefined} & \text{otherwise}
    \end{cases}
  $}
\vspace{-3mm}
\caption{Typing Context Extension}
\vspace{-3mm}
\label{fig:ctx-ext}
\end{figure}

\begin{remark}[Empty Type]
  \upshape
  A refinement type may be \emph{empty}, with no inhabited member.

  We can construct such a type under the empty context $\emptyctx$ as
  $(\vb{x}{S}\esetof{\efalse})$ with any base types $\dte S$.
  A more specific example is a refinement type for an integer that is both
  negative and positive
  $(\vb{x}{\tint}\esetof{x > 0 \land x < 0})$.
  Similarly, under the context $\vb{x^\omega}{\tint\esetof{x > 0}}$, the
  refinement type $\vb{y}{\tint\esetof{y < 0 \land y > x}}$ is empty.
  In these cases, the typing context with the specified type becomes
  inconsistent, i.e.\ the encoded context gives a proof of falsity.

  Moreover, an empty type can also occur without inconsistency.
  For instance, in a typing context of $\vb{x^0}{\tint}$, the type
  $\vb{y}{\tint\esetof{y > x}}$ is empty --- it is not possible to produce such
  a value without referring to $\dexp x$ (cf.\ \ruleTEVar).

  \label{rem:empty-value-type}
\end{remark}
\subsection{Endpoint Projection: From Global Contexts and Types to Local Contexts and Types}
\label{subsection:theory-projection}

In the methodology of multiparty session types, developers specify a global
type, and obtain local types for the participants via \emph{endpoint
projection} (\emph{projection} for short).
In the original theory,
projection is a \emph{partial} function that takes a global type $\dgt G$ and a
participant $\ppt p$, and returns a local type $\dtp L$.
The resulting local type $\dtp L$ describes a the local communication behaviour
for participant $\ppt p$ in the global scenario.
Such workflow has the advantage that each endpoint can obtain a local type
separately, and implement a process of the given type, hence providing
modularity and scalability.

Projection is defined as a \emph{partial} function, since only
\emph{well-formed}
global types can be projected to all participants.
In particular, a \emph{partial} merge operator $\sqcup$ is used during the
projection, for creating a local type
$\Sigma \vdash \dtp{L_1} \sqcup \dtp{L_2} = \dtp{L_{\text{merged}}}$
that captures the behaviour of two local types, under context $\Sigma$.

In RMPST, we first define the projection of global typing contexts
(\cref{fig:proj-gctx}), and then define the projection of global types under
a global typing context (\cref{fig:proj-gty}).
We use expression typing judgements in the definition of projection, to
type-check expressions against their prescribed types.

\mypara{Projection of Global Contexts}
We define the judgement $\ctxproj{\Gamma}{p} = \Sigma$ for the projection of
global typing context $\Gamma$ to participant $\ppt p$ in \cref{fig:proj-gctx}.
In the global context $\Gamma$, a variable $\dexp x$ is annotated with the
set of participants $\bigP$ who know the value.
If the projected participant $\ppt p$ is in the set $\bigP$, \rulePVarOmega\
is applied to obtain an unrestricted variable in the resulting local context;
Otherwise, \rulePVarZero\  is applied to obtain an irrelevant variable.

\mypara{Projection of Global Types with a Global Context}
When projecting a global type $\dgt G$, we include a global context $\Gamma$,
forming a judgement of form $\gtctxproj{\Gamma}{G}{p} = \ltctx{\Sigma}{L}$.
Projection rules are shown in \cref{fig:proj-gty}.
Including a typing context allows us to type-check expressions during
projection, hence ensuring that variables attached to recursive protocols are
well-typed.

\input{fig/rules/proj-gctx.tex}

\input{fig/rules/proj-gty.tex}

If the prefix of $\dgt G$ is a message from role $\ppt p$ to role $\ppt q$,
the projection results a local type with a send (resp.\ receive) prefix into
role $\ppt p$ (resp.\ $\ppt q$) via \rulePSend\ (resp.\ \rulePRecv).
For other roles $\ppt r$, the projection results in a local type with a
\emph{silent label} via \rulePPhi, with prefix $\tphi{l}{x}{T}{}$
This follows the concept of a coordinated distributed system, where all the
processes follow a global protocol, and base assumptions of their local
actions on actions of other roles not involving them.
The projection defined in the original MPST theory does not contain information
for role $\ppt r$ about a message between $\ppt p$ and $\ppt q$.
We use the silent prefix to retain such information, especially the
refinement type $\dte T$ of the payload.
For merging two local types (as used in \rulePPhi), we use a simple plain merge
operator defined as $ \Sigma \vdash \dtp L \sqcup \dtp L = \dtp L $, requiring
two local types to be identical in order to be merged.\footnotemark
\footnotetext{We build upon the standard MPST theory with plain merging.
Full merge~\cite{LMCS12Parameterised}, allowing certain different index sets to
be merged, is an alternative, more permissive merge operator. Our
implementation \Ourtool uses the more permissive merge operator for better
expressiveness.}

If the prefix of $\dgt G$ is a recursive protocol
$\gtrecur{t}{\vb{x}{T}}{\dexp{x:=E}}{G}$, the projection preserves the
recursion construct if
the projected role is in the inner protocol via \rulePRecOne\ and that the
expression $\dexp E$ can be typed with type $\dte T$ under the projected local context.
Typing expressions under local contexts ensures that no irrelevant
variables $\dexp{x^0}$ are used in the expression $\dexp E$, as no typing derivation
exists for irrelevant variables.
Otherwise projection results in $\tend$ via \rulePRecTwo.
If $\dgt G$ is a type variable $\gtvar{t}{\dexp{x:=E}}$, we similarly
validate that the expression $\dexp E$ carries the specified type in the
correspondent recursion definition, and its projection also preserves the type
variable construct.

\begin{example}[Projection of Global Types of \cref{example:gty} (1)]
  We draw attention to the projection of $\dgt{G_1}$ to $\ppt C$, under the empty context $\emptyctx$.
  \[
    \gtctxproj{\emptyctx}{G_1}{C} =
    \ltctx{\emptyctx}{
      \tphi{Fst}{x}{\tint}{}
      \toffersingle{B}{\dlbl{Snd}(\vb{y}{\tint\esetof{x = y}}).
        \ttakesingle{D}{\dlbl{Trd}(\vb{z}{\tint\esetof{x = z}}).
        \tend}
      }
    }
  \]
  We note that the local type for $\ppt C$ has a silent prefix
  $\dlbl{Fst}(\vb{x}{\tint})$, which binds the variable $\dexp x$
  in the continuation.
  The silent prefix adds the variable $\dexp x$ and its type to the ``local
  knowledge'' of the endpoint $\ppt C$, yet the actual value of $\dexp x$
  is unknown.
  \label{example:proj}
\end{example}
\begin{remark}[Empty Session Type]
  \upshape
  Global types $\dgt G$ and local types $\dtp L$ can be empty because one of
  the value types in the protocol in an empty type (cf.\
  \cref{rem:empty-value-type}).

  For example, the local type
  $\ttakesingle{A}{\dlbl{Impossible}(\vb{x}{\tint\esetof{x > 0 \land x < 0}}).\tend}$
  cannot be implemented, since such an $\dexp x$ cannot be provided.

  For the same reason, the local type
  $\tphi{Pos}{x}{\tint\esetof{x >
      0}}{}\ttakesingle{A}{\dlbl{Impossible}(\vb{y}{\tint\esetof{y > x}}).\tend}
  $ cannot be implemented.
  \label{rem:empty-session-type}
\end{remark}
\begin{remark}[Implementable Session Types]
  \upshape
  Consider the following session type:
  $$\dtp L =
  \toffersingle{B}{\dlbl{Num}(\vb{x}{\tint}).
    \ttake{B}{\begin{array}{@{}l@{}}
        \dlbl{Pos}({\tunit\esetof{x > 0}}).\tend\\
        \dlbl{Neg}({\tunit\esetof{x < 0}}).\tend
      \end{array}
   }}.
  $$
  When the variable $\dexp x$ has the value $\eintlit{0}$, neither of the
  choices $\dlbl{Pos}$ or $\dlbl{Neg}$ could be selected, as the refinements
  are not satisfied.
  In this case, the local type $\dtp L$ cannot be implemented, as the internal choice
  callback may not be implemented in a \emph{total} way, i.e.\ the callback
  returns a choice label for all possible inputs of integer $\dexp
  x$.\footnotemark
\footnotetext{Since we use a permissive \keyword{ML} effect in the callback
  type, allowing all side effects to be performed in the callback, the
  callback may throw exceptions or diverge in case of unable to return a value.
}
\end{remark}

\subsection{Labelled Transition System (LTS) Semantics}
\label{subsection:theory-semantics}

We define the labelled transition system (LTS) semantics for global types and
local types.
We show the trace equivalence of a global type and the collection of local
types projected from the global type, to demonstrate that projection preserves
LTS semantics.
The equivalence result allows us to use the projected local types for the
implementation of local roles separately.
Therefore, we can implement the endpoints in \fstar separately, and they compose
to the specified protocol.

We also prove a type safety result that well-formed global types cannot be
stuck. This, combined with the trace equivalence result, guarantees that
endpoints are free from deadlocks.

\mypara{Actions}
We begin with defining actions in the LTS system.
We define the label in the LTS as $\alpha \bnfas \ltsmsg{p}{q}{l}{x}{T}$,
a message from role $\ppt p$ to $\ppt q$ with label $\dlbl l$ carrying a
value named $\dexp x$ with type $\dte T$.
We define $\subj{\alpha} = \psetof{p, q}$ to be the subjects of the action
$\alpha$, namely the two roles in the action.

\mypara{Semantics of Global Types}
We define the LTS semantics of global types in \cref{fig:lts-gty}.
Different from the original LTS semantics in \cite{ICALP13CFSM}, we include the
context $\Gamma$ in the semantics along with the global type $\dgt G$.
Therefore, the judgements of global LTS reduction have form
${\gtctx{\Gamma}{G} \stepsto[\alpha] \gtctx{\Gamma'}{G'}}$.

\ruleGPfx\ allows the reduction of the prefix action in a global type.
An action, matching the definition in set defined in the prefix, allows
the correspondent continuation to be selected.
The resulting global type is the matching continuation and the resulting
context contains the variable binding in the action.

\ruleGCtx\ allows the reduction of an action that is causally independent
of the prefix action in a global type, here, the subjects of the action are
disjoint from the prefix of the global type.
If all continuations in the global types can make the reduction of that action
to the same context, then the result context is that context and the result
global type is one with continuations after reduction.
When reducing the continuations, we add the variable of
the prefix action into the context, but tagged with an empty set of known roles.
This addition ensures that relevant information obtainable from the prefix message
is not lost when performing reduction.

\ruleGRec\ allows the reduction of a recursive type by unfolding the type once.

\input{fig/rules/lts-gty.tex}

\begin{example}[Global Type Reductions]
  We demonstrate two reduction paths for a global type
  $$\dgt G = \gtbransingle{p}{q}{\dlbl{Hello}(\dexp{x}: \tint\esetof{x <
      0}).\gtbransingle{r}{s}{\dlbl{Hola}(\dexp{y}: \tint\esetof{y >
        x}).\gtend}}.$$
  Note that the two messages are not causally related (they have disjoint
  subjects).
  We have the following two reduction paths of $\gtctx{\emptyctx}{G}$ (omitting
  payload in LTS actions):
  \[
    \begin{array}{rl}
      & \gtctx{\emptyctx}{G} \\
      \text\ruleGPfx\stepsto[\ltsmsgsimpl{p}{q}{Hello}] &
      \gtctx{\vb{{x}^{\psetof{p, q}}}{\tint\esetof{x <
            0}}}{\gtbransingle{r}{s}{\dlbl{Hola}(\vb{y}{\tint\esetof{y >
              x}}).\gtend}} \\
      \text\ruleGPfx\stepsto[\ltsmsgsimpl{r}{s}{Hola}] &
      \gtctx{\vb{{x}^{\psetof{p, q}}}{\tint\esetof{x < 0}}, \vb{{y}^{\psetof{r,
              s}}}{\tint\esetof{y > x}}}{\gtend} \\[2.5ex]
      & \gtctx{\emptyctx}{G} \\
      \text\ruleGCtx\stepsto[\ltsmsgsimpl{r}{s}{Hola}] &
      \gtctx{\vb{{x}^{\varnothing}}{\tint\esetof{x < 0}}, \vb{{y}^{\psetof{r,
              s}}}{\tint\esetof{y >
            x}}}{\gtbransingle{p}{q}{\dlbl{Hello}(\vb{x}{\tint\esetof{x <
              0}}).\gtend}} \\
      \text\ruleGPfx\stepsto[\ltsmsgsimpl{p}{q}{Hello}] &
      \gtctx{\vb{{x}^{\psetof{p, q}}}{\tint\esetof{x < 0}}, \vb{{y}^{\psetof{r,
              s}}}{\tint\esetof{y > x}}}{\gtend}
    \end{array}
  \]
\end{example}

\mypara{Semantics of Local Types}
We define the LTS semantics of local types in \cref{fig:lts-lty}.
Similar to global type LTS semantics, we include the local context $\Sigma$
in the semantics.
Therefore, the judgements of local LTS reductions have form
${\ltctx{\Sigma}{L} \stepsto[\alpha] \ltctx{\Sigma'}{L'}}$.
When defining the LTS semantics, we also use judgements of form
${\ltctx{\Sigma}{L} \stepsto[\epsilon] \ltctx{\Sigma'}{L'}}$.
It represents a silent action that can occur without an observed action.
We write $\stepsto[\epsilon]^*$
to denote the reflexive transition closure of silent actions $\stepsto[\epsilon]$.

We first have a look at silent transitions.
\ruleEPhi\ allows the variable in a silent type to be added into the local
context in the irrelevant form.
This rule allows local roles to obtain knowledge from the
messages in the global protocol without their participation.

\ruleECtx\ allows prefixed local type to make a silent transition, if all of
its continuations are allowed to make a silent transition to reach the same
context.
The rule allows a prefixed local type to obtain new knowledge
about irrelevant variables, if such can be obtained in all possible
continuations.

\ruleERec\ unfolds recursive local types, analogous to the unfolding of
global types.

For concrete transitions, we have \ruleLSend\ (resp.\ \ruleLRecv) to
reduce a local type with a sending (resp.\ receiving) prefix, if the action
label is in the set of labels in the local type.
The resulting context contains the variable in the message as a concrete
variable, since the role knows the value via communication.
The resulting local type is the continuation corresponding to the action label.

In addition, \ruleLEps\ permits any number of silent actions to be taken before a
concrete action.

\input{fig/rules/lts-lty.tex}

\begin{remark}[Reductions for Empty Session Types]
  \upshape
  We consider empty session types to be reducible, since it is not possible to
  distinguish which types are inhabited.
  However, it does not invalidate the safety properties of endpoints, since
  no such endpoints can be implemented for an empty session type.

\end{remark}

\mypara{Relating Semantics of Global and Local Types}
We extend the LTS semantics to a collection of local types in
\cref{def:lts-collection}, in order to prove that projection preserves
semantics.
We define the semantics in a synchronous fashion.

The set of local types reduces with an action $\alpha =
\ltsmsg{p}{q}{l}{x}{T}$, if the local type for role $\ppt p$ and $\ppt q$
both reduce with that action $\alpha$.
All other roles in the set of the local types are permitted to make silent
actions ($\epsilon$ actions).

Our definition deviates from the standard definition
\cite[Def. 3.3]{ICALP13CFSM} in two ways:
One is that we use a synchronous semantics, so that one action involves
two reductions, namely at the sending and receiving sides.
Second is that we use contexts and silent transitions in the LTS semantics.
The original definition requires all non-action roles to be identical, whereas
we relax the requirement to allow silent transitions.

\begin{definition}[LTS over a collection of local types]
  A configuration $s = \setof{\ltctx{\Sigma_{\ppt{r}}}{L_{\ppt{r}}}}_{\ppt r \in \bigP}$
  is a collection of local types and contexts, indexable via participants.

  Let $\ppt p \in \bigP$ and $\ppt q \in \bigP$.
  We say $s = \setof{\ltctx{\Sigma_{\ppt{r}}}{L_{\ppt{r}}}}_{\ppt r \in \bigP}
  \stepsto[\alpha = \ltsmsg{p}{q}{l}{x}{T}]
  s' = \setof{\ltctx{\Sigma'_{\ppt{r}}}{L'_{\ppt{r}}}}_{\ppt r \in \bigP}$ if
  \begin{enumerate}
    \item $\ltctx{\Sigma_{\ppt p}}{{L}_{\ppt p}} \stepsto[\alpha] \ltctx{\Sigma'_{\ppt p}}{L'_{\ppt p}}$ and,
          $\ltctx{\Sigma_{\ppt q}}{{L}_{\ppt q}} \stepsto[\alpha] \ltctx{\Sigma'_{\ppt q}}{L'_{\ppt q}}$ and,
    \item for all $\ppt s \in \bigP, \ppt s \neq \ppt p, \ppt s \neq \ppt q$.
    $\ltctx{\Sigma_{\ppt s}}{L_{\ppt s}} \ltequivmany \ltctx{\Sigma'_{\ppt s}}{L'_{\ppt s}}$
  \end{enumerate}
  \label{def:lts-collection}
\end{definition}

For a closed global type $\dgt G$ under context $\Gamma$, we show that the
global type makes the same trace of reductions as the collection of local
types obtained from projection.
We prove it in \cref{thm:trace-eq}.

\begin{definition}[Association of Global Types and Configurations]
  Let $\gtctx{\Gamma}{G}$ be a global context.

  The collection of local contexts \emph{associated} to $\gtctx{\Gamma}{G}$,
  is defined as the configuration \linebreak
  $\setof{\gtctxproj{\Gamma}{G}{r}}_{\ppt r \in \dgt G}$.
  We write $s \ltassoc \gtctx{\Gamma}{G}$ if a configuration $s$ is the associated to
  $\gtctx{\Gamma}{G}$.
\end{definition}

\begin{theorem}[Trace Equivalence]
  Let $\gtctx{\Gamma}{G}$ be a closed global context and
  $s \ltassoc \gtctx{\Gamma}{G}$
  be a configuration associated with the global context.

  $\gtctx{\Gamma}{G} \stepsto[\alpha] \gtctx{\Gamma'}{G'}$ if and only if $s \stepsto[\alpha] s'$,
  where $s' \ltassoc \gtctx{\Gamma'}{G'}$.
  \label{thm:trace-eq}
\end{theorem}

The theorem states that semantics are preserved after projection.
Practically, we can implement local processes separately, and run them
in parallel with preserved semantics.

We also show that a well-formed global type $\dgt G$ has progress.
This means that a well-formed global type does not get \emph{stuck}, which
implies deadlock freedom.

\begin{definition}[Well-formed Global Types]
  A global type under typing context $\gtctx{\Gamma}{G}$ is well-formed, if
  \begin{enumerate*}
    \item $\dgt G$ does not contain free type variables,
    \item $\dgt G$ is contractive \cite[\S 21]{PierceTAPL}, and
    \item for all roles in the protocol $\ppt r \in \dgt G$, the
      projection $\gtctxproj{\Gamma}{G}{r}$ is defined.
  \end{enumerate*}

  We also say a global type $\dgt G$ is well-formed, if
  $\gtctx{\varnothing}{G}$ is well-formed.
\end{definition}
\begin{theorem}[Preservation of Well-formedness]
  If $\gtctx{\Gamma}{G}$ is a well-formed global type under typing context, and
  $\gtctx{\Gamma}{G} \stepsto[\alpha] \gtctx{\Gamma'}{G'}$, then
  $\gtctx{\Gamma'}{G'}$ is well-formed.
  \label{thm:wf-preservation}
\end{theorem}
\begin{definition}[Progress]
  A configuration $s$ satisfies progress, if either
  \begin{enumerate*}
    \item For all participants $\ppt p \in s$, $\dtp{\Lp} = \tend$, or
    \item there exists an action $\alpha$ and a configuration $s'$ such that $s
      \stepsto[\alpha] s'$.
  \end{enumerate*}

  A global type under typing context $\gtctx{\Gamma}{G}$ satisfies progress, if
  its associated configuration $s \ltassoc \gtctx{\Gamma}{G}$, exists and satisfies progress.

  We also say a global type $\dgt G$ satisfies progress, if
  $\gtctx{\varnothing}{G}$ satisfies progress.
\end{definition}
\begin{theorem}[Progress]
  If $\gtctx{\Gamma}{G}$ is a well-formed global type under typing context,
  then $\gtctx{\Gamma}{G}$ satisfies progress.
  \label{thm:progress}
\end{theorem}

\begin{theorem}[Type Safety]
  If $\dgt G$ is a well-formed global type, then for any global type under
  typing context $\gtctx{\Gamma'}{G'}$ such that $\gtctx{\varnothing}{G}
  \stepstomany[] \gtctx{\Gamma'}{G'}$, $\gtctx{\Gamma'}{G'}$ satisfies progress.
  \label{thm:type-safety}
\end{theorem}
\begin{proof}
  Direct consequence of \cref{thm:wf-preservation} and \cref{thm:progress}.
\end{proof}


%% file: fig/rules/typing-exp.tex
\begin{figure}
  \begin{mathpar}
    \inferrule[\ruleWfRty]{
      \Sigma^+, \vb{x}{S} \vdash \vb{E}{\tbool}
    }{
      \Sigma \vdash (\vb{x}{S}\esetof{E}) ~ty
    }

    \inferrule[\ruleTEVar]{
    }{
        \Sigma_1, \vb{x^\omega}{T}, \Sigma_2 \vdash \dexp x : \dte T
    }

    \inferrule[\ruleTEPlus]{
        \Sigma \vdash \dexp{E_1} : \tint
        \and
        \Sigma \vdash \dexp{E_2} : \tint
    }{
        \Sigma \vdash \dexp{E_1 + E_2} : (\vb{v}{\tint}\esetof{v = E_1 + E_2})
    }
    \vspace{-3mm}

    \inferrule[\ruleTESub]{
        \Sigma \vdash \dexp{E} : (\vb{v}{S}\esetof{E_1})
        \and
        \valid{\enc{\Sigma} \land \enc{\dexp{E_1}} \implies \enc{\dexp{E_2}}}
    }{
        \Sigma \vdash \dexp{E} : (\vb{v}{S}\esetof{E_2})
    }

    \inferrule[\ruleTEConst]{
    }{
        \Sigma \vdash \dexp{\underline{n}} : (\vb{v}{\tint}\esetof{v = \underline{n}})
    }
  \end{mathpar}
  \vspace{-3mm}
  \caption{Selected Typing Rules for Expressions in a Local Typing Context}
  \vspace{-3mm}
  \label{fig:typing-expression}
\end{figure}

%% file: fig/rules/proj-gctx.tex
\begin{figure}
\begin{mathpar}
  \boxed{\ctxproj{\Gamma}{p} = \Sigma}

  \inferrule[\rulePEmpty]
  {
  }
  {
    \ctxproj{\emptyctx}{p} = \emptyctx
  }

  \inferrule[\rulePVarOmega]
  {
    \ppt p \in \bigP
    \and
    \ctxproj{\Gamma}{p} = \Sigma
  }
  {
    \ctxproj{\ctxc{\Gamma}{x^{\bigP}}{T}}{p}
    =
    \ctxc{\Sigma}{x^\omega}{T}
  }

  \inferrule[\rulePVarZero]
  {
    \ppt p \notin \bigP
    \and
    \ctxproj{\Gamma}{p} = \Sigma
  }
  {
    \ctxproj{\ctxc{\Gamma}{x^{\bigP}}{T}}{p}
    =
    \ctxc{\Sigma}{x^0}{T}
  }
\end{mathpar}
\vspace{-5mm}
\caption{Projection Rules for Global Contexts}
\vspace{-5mm}
\label{fig:proj-gctx}
\end{figure}

%% file: fig/rules/proj-gty.tex
\begin{figure}
\begin{mathpar}
  \myinferrule[\rulePSend]
  {
    \ctxproj{\Gamma}{p} = \Sigma
    \and
    \forall i \in I.
    \and
    \Sigma \vdash \dte{T_i} ~ty
    \and
    \gtctxproj{\ctxext{\Gamma}{x_{i}^{\psetof{p, q}}}{T_i}}{G_i}{p} = \ltctx{\Sigma_i}{L_i}
  }
  {
    \gtctxproj{\Gamma}{(\gtbran{p}{q}{\dlbl{l_i}(\vb{x_i}{T_i}).G_i}_{\color{black} i \in I})}{p}
    =
    \ltctx{\Sigma}{\ttake{q}{\dlbl{l_i}(\vb{x_i}{T_i}).L_i}_{i \in I} }
  }

  \vspace{-3mm}

  \myinferrule[\rulePRecv]
  {
    \ctxproj{\Gamma}{q} = \Sigma
    \and
    \forall i \in I.
    \and
    \Sigma \vdash \dte{T_i} ~ty
    \and
    \gtctxproj{\ctxext{\Gamma}{x_{i}^{\psetof{p, q}}}{T_i}}{G_i}{q} = \ltctx{\Sigma_i}{L_i}
  }
  {
    \gtctxproj{\Gamma}{(\gtbran{p}{q}{\dlbl{l_i}(\vb{x_i}{T_i}).G_i}_{\color{black} i \in I})}{q}
    =
    \ltctx{\Sigma}{\toffer{p}{\dlbl{l_i}(\vb{x_i}{T_i}).L_i}_{i \in I} }
  }

  \vspace{-3mm}

  \myinferrule[\rulePPhi]
  {
    \ctxproj{\Gamma}{r} = \Sigma
    \\
    \ppt r \notin \psetof{p, q}
    \\
    \forall i \in I.
    \\
    \Sigma \vdash \dte{T_i} ~ty
    \\
    \gtctxproj{\ctxext{\Gamma}{x_{i}^{\psetof{p, q}}}{T_i}}{G_i}{r} = \ltctx{\Sigma_i}{L_i}
  }
  {
    \gtctxproj{\Gamma}{(\gtbran{p}{q}{\dlbl{l_i}(\vb{x_i}{T_i}).G_i}_{\color{black} i \in I})}{r}
    =
    \ltctx{\Sigma}{\sqcup_{i \in I} \tphi{l_i}{x_i}{T_i}{L_i}}
  }
  \vspace{-3mm}

  \myinferrule[\rulePRecOne]
  {
    \ctxproj{\Gamma}{r} = \Sigma
    \\
    \ppt r \in \dgt G
    \\
    \gtctxproj{\ctxext{\Gamma}{x^{\setof{\ppt r | \ppt r \in \dgt G}}}{T}}{G}{r} = \ltctx{\Sigma'}{L}
    \\
    \Sigma \vdash \dte{T} ~ty
    \\
    \Sigma \vdash \dexp E \eoft \dte T
  }
  {
    \gtctxproj{\Gamma}{\gtrecur{t}{\vb{x}{T}}{\dexp{x := E}}{G})}{r}
    =
    \ltctx{\Sigma}{\trecur{t}{}{\vb{x}{T}}{\dexp{x := E}}{\dtp L}}
  }
  \vspace{-3mm}

  \myinferrule[\rulePRecTwo]
  {
    \ctxproj{\Gamma}{r} = \Sigma
    \and
    \ppt r \notin \dgt G
  }
  {
    \gtctxproj{\Gamma}{\gtrecur{t}{\vb{x}{T}}{\dexp{x := E}}{G})}{r}
    =
    \ltctx{\Sigma}{\tend}
  }

  \myinferrule[\rulePEnd]
  {
    \ctxproj{\Gamma}{r} = \Sigma
  }
  {
    \gtctxproj{\Gamma}{\gtend}{r}
    =
    \ltctx{\Sigma}{\tend}
  }
  \vspace{-3mm}

  \boxed{\gtctxproj{\Gamma}{G}{p} = \ltctx{\Sigma}{L}}

  \myinferrule[\rulePVar]
  {
    \ctxproj{\Gamma}{r} = \Sigma = \Sigma_1, \vb{x}{T}, \Sigma_2
    \and
    \Sigma_1, \vb{x}{T}, \Sigma_2 \vdash \vb{E}{T}
  }
  {
    \gtctxproj{\Gamma}{\gtvar{t}{\dexp{x := E}}}{r}
    =
    \ltctx{\Sigma}{\tvar{t}{\dexp{x := E}}}
  }
  \vspace{-3mm}

\end{mathpar}
\vspace{-3mm}
\caption{Projection Rules for Global Types}
\vspace{-3mm}
\label{fig:proj-gty}
\end{figure}

%% file: fig/rules/lts-gty.tex
\begin{figure}
\begin{mathpar}
  \myinferrule[\ruleGPfx]
    {j \in I}
    {
      \gtctx{\Gamma}{\gtbran{p}{q}{\dlbl{l_i}(\vb{x_i}{T_i}).G_i}_{i \in I}}
      \stepsto[\ltsmsg{p}{q}{l_j}{x_j}{T_j}]
      \gtctx{\ctxext{\Gamma}{x_{j}^{\psetof{p, q}}}{T_j}}{\dgt{G_j}}
    }
  \vspace{-3mm}

  \myinferrule[\ruleGCtx]
    {
      \psetof{p, q} \cap \subj{\alpha} = \varnothing
      \and
      \forall j \in I.
      \gtctx{\ctxext{\Gamma}{x_{j}^{\varnothing}}{T_j}}{\dgt{G_j}}
        \stepsto[\alpha]
      \gtctx{\Gamma'}{\dgt{G_j'}}
    }
    {
      \gtctx{\Gamma}{\gtbran{p}{q}{\dlbl{l_i}(\vb{x_i}{T_i}).G_i}_{i \in I}}
      \stepsto[\alpha]
      \gtctx{\Gamma'}{\gtbran{p}{q}{\dlbl{l_i}(\vb{x_i}{T_i}).G_i'}_{i \in I}}
    }
  \vspace{-3mm}

  \boxed{\gtctx{\Gamma}{G} \stepsto[\alpha] \gtctx{\Gamma'}{G'}}

  \myinferrule[\ruleGRec]
    {
      \gtctx{\ctxext{\Gamma}{x^{\setof{\ppt r | \ppt r\in \dgt G}}}{T}}{G\subst{\gtrecursimpl{t}{\vb{x}{T}}{G}}{\mathbf{t}}}
      \stepsto[\alpha]
      \gtctx{\Gamma'}{G'}
    }
    {
      \gtctx{\Gamma}{\gtrecur{t}{\vb{x}{T}}{\dexp{ x := E}}{G}}
      \stepsto[\alpha]
      \gtctx{\Gamma'}{G'}
    }

  \vspace{-3mm}
\end{mathpar}
\vspace{-3mm}
\caption{LTS Semantics for Global Types}
\vspace{-3mm}
\label{fig:lts-gty}
\end{figure}

%% file: fig/rules/lts-lty.tex
\begin{figure}
\begin{mathpar}

    \myinferrule[\ruleECtx]
    {
      \dagger \in \setof{\&, \oplus}
      \and
      \forall j \in I.
      \and
      \ltctx{\ctxext{\Sigma}{x_j^0}{T_j}}{L_j}
      \ltequiv
      \ltctx{\Sigma'}{L_j'}
    }
    {
      \ltctx{\Sigma}{\tgeneric{q}{\dlbl{l_i}(\vb{x_i}{T_i}).L_i}_{i \in I}}
      \ltequiv
      \ltctx{\Sigma'}{\tgeneric{q}{\dlbl{l_i}(\vb{x_i}{T_i}).L_i'}_{i \in I}}
    }

    \myinferrule[\ruleERec]
      {
      }
      {
        \ltctx{\Sigma}{\trecur{t}{}{\vb{x}{T}}{\dexp{ x := E}}{L}}
        \ltequiv
        \ltctx{\ctxext{\Sigma}{x^\omega}{T}}{L\subst{\trecursimpl{t}{}{\vb{x}{T}}{L}}{\mathbf{t}}}
      }

    \boxed{\ltctx{\Sigma}{L} \stepsto[\epsilon] \ltctx{\Sigma'}{L'}}

    \myinferrule[\ruleEPhi]
    {
    }
    {
      \ltctx{\Sigma}{\tphi{l}{x}{T}{L}}
      \ltequiv
      \ltctx{\ctxext{\Sigma}{x^0}{T}}{L}
    }
\end{mathpar}
\hrule

\begin{mathpar}
  \myinferrule[\ruleLSend]
  {
    j \in I
  }
  {
    \ltctx{\Sigma}{\ttake{q}{\dlbl{l_i}(\vb{x_i}{T_i}).L_i}_{i \in I}}
    \stepsto[\ltsmsg{p}{q}{l_j}{x_j}{T_j}]
    \ltctx{\ctxext{\Sigma}{x_j^\omega}{T_j}}{L_j}
  }
  \vspace{-3mm}

  \myinferrule[\ruleLRecv]
  {
    j \in I
  }
  {
    \ltctx{\Sigma}{\toffer{p}{\dlbl{l_i}(\vb{x_i}{T_i}).L_i}_{i \in I}}
    \stepsto[\ltsmsg{p}{q}{l_j}{x_j}{T_j}]
    \ltctx{\ctxext{\Sigma}{x_j^\omega}{T_j}}{L_j}
  }
  \vspace{-3mm}

  \boxed{\ltctx{\Sigma}{L} \stepsto[\alpha] \ltctx{\Sigma'}{L'}}

  \myinferrule[\ruleLEps]
    {
      \ltctx{\Sigma}{L}
      \ltequiv
      \ltctx{\Sigma''}{L''}
      \and
      \ltctx{\Sigma''}{L''}
      \stepsto[\alpha]
      \ltctx{\Sigma'}{L'}
    }
    {
      \ltctx{\Sigma}{L}
      \stepsto[\alpha]
      \ltctx{\Sigma'}{L'}
    }
  \vspace{-3mm}
\end{mathpar}
\vspace{-3mm}
\caption{LTS Semantics for Local Types}
\vspace{-3mm}
\label{fig:lts-lty}
\end{figure}

%% file: evaluation.tex
\section{Evaluation}
\label{section:evaluation}

We evaluate the expressiveness and performance of our toolchain \Ourtool.
We describe the methodology and setup (\cref{subsection:eval-method}), and
comment on the compilation time (\cref{subsection:eval-comp}) and the execution
time (\cref{subsection:eval-exec}).
We demonstrate the expressiveness of \Ourtool (\cref{subsection:eval-expressive}) by
implementing examples from the session type literature and comparing with related
work.
The source files of the benchmarks used in this section are included in our
artifact, along with a script to reproduce the results.

\subsection{Methodology and Setup}
\label{subsection:eval-method}

We measure the time to generate the CFSM representation from a \scrib
protocol (\emph{CFSM}), and the time to generate \fstar
code from the CFSM representation (\emph{\fstar APIs}). Since
the generated APIs in \fstar need to be type-checked before use, we
also measure the type-checking time for the generated code
(\emph{Gen.\@ Code}). Finally, we provide a simple
implementation of the callbacks and measure the type-checking time for
the callbacks against the generated type (\emph{Callbacks}).

To execute the protocols, we need a network transport to connect the participants by providing
appropriate sending and receiving primitives.
In our experiment setup, we use the standard library module \code{FStar.Tcp}
to establish TCP connections between participants, and provide a simple
serialisation module for base types. Due to the small size of our
payloads, we set \code{TCP\_NODELAY} to avoid the delays introduced by
the congestion control algorithms.
Since our entry point to execute the protocol is parameterised by the connection/transport type,
the implementation may use other connections if developers wish, e.g.\ an
in-memory queue for local setups.
We measure the execution time of the protocol (\emph{Execution Time}).

\begin{figure}
    \centering
\begin{lstlisting}[language=Scribble, numbers=none]
global protocol PingPong£${}_n$£(role A, role B) {
  choice at A { Ping(x£${}_1$£:int) from A to B;             Pong(y£${}_1$£:int) from B to A; @"y£${}_1$£>x£$_1$£"
                Ping(x£${}_2$£:int) from A to B; @"x£${}_2$£>y£$_1$£"    Pong(y£${}_2$£:int) from B to A; @"y£${}_2$£>x£$_2$£"
                £$\cdots$£
                Ping(x£${}_n$£:int) from A to B; @"x£${}_n$£>y£$_{n-1}$£" Pong(y£${}_n$£:int) from B to A; @"y£${}_n$£>x£$_n$£"
                do PingPong£${}_n$£(A, B); }
           or { Bye() from A to B;                           Bye() from B to A; } }
\end{lstlisting}
    \vspace{-3mm}
    \caption{Ping Pong Protocol (Parameterised by Protocol Length $n$)}
    \label{fig:pingpong}
    \vspace{-3mm}
\end{figure}
To measure the overhead of our implementation, we compare against an
implementation of the protocol without session types or refinement types, which
we call \emph{bare implementation}.
In this implementation, we use the same sending and receiving primitives (i.e.\
\lstinline+connection+) as in the toolchain implementation.
The bare implementation is in a series of direct calls of sending and receiving
primitives, for the same communication pattern, but without the generated APIs.

We use a Ping Pong protocol (\cref{fig:pingpong}), parameterised by
the protocol length, i.e.\ the number of Ping Pong messages $n$ in a protocol
iteration.
When the protocol length $n$ increases, the number of CFSM states increases
linearly, which gives rise to longer generated code and larger generated types.
In each Ping Pong message, we include payload of increasing numbers, and encode
the constraints as protocol refinements.

We study its effect on the compilation time (\cref{subsection:eval-comp}) and
the execution time (\cref{subsection:eval-exec}).
We run the experiment on varying sizes of $n$, up to 25.
Larger sizes of $n$ leads to unreasonably large resource usage during
type-checking in \fstar.
\cref{tab:run-time} reports the results for the Ping Pong protocol
in \cref{fig:pingpong}.

We run the experiments under a network of latency of 0.340ms (64 bytes ping), and
repeat each experiment 30 times.
Measurements are taken using a machine with
Intel i7-7700K CPU (4.20 GHz, 4 cores, 8 threads), 16 GiB RAM, operating system Ubuntu
18.04, OCaml compiler version 4.08.1, \fstar compiler commit
\code{\href{https://github.com/FStarLang/FStar/commit/8040e34a2c6031276fafd2196b412d415ad4591a}{8040e34a}},
Z3 version 4.8.5.

\subsection{Compilation Time}
\label{subsection:eval-comp}
\mypara{CFSM and \fstar Generation Time}
We measure the time taken for \scrib to generate the CFSM from the
protocol in \cref{fig:pingpong}, and for the code generation tool to
convert the CFSM to \fstar APIs.
We observe from \cref{tab:run-time} that the generation time
for CFSMs and \fstar APIs is short.
It takes less than 1 second to complete the generation
phase for each case.

\begin{table}
\centering
    \begin{tabular}{|c|r|r|r|r|c|}
        \hline
       Protocol         &
        \multicolumn{2}{c|}{Generation Time}  &
        \multicolumn{2}{c|}{Type Checking Time} &
         Execution Time \\
         \cline{2-5}
         Length ($n$)  &  CFSM &  \fstar APIs & Gen.\@ Code & Callbacks & (100 000 ping-pongs)  \\
        \hline
        \hline
        bare & n/a & n/a &  n/a & n/a &  28.79s \\
         \hline
        1  & 0.38s & 0.01s & 1.28s & 0.34s & 28.75s \\
        \hline
        5  & 0.48s & 0.01s & 3.81s & 1.12s & 28.82s \\
        \hline
        10 & 0.55s & 0.01s & 14.83s & 1.34s & 28.84s \\
        \hline
        15 & 0.61s & 0.01s & 42.78s & 1.78s & n/a \\
        \hline
        20 & 0.69s & 0.02s & 98.35s & 2.54s & 28.81s \\
        \hline
        25 & 0.78s & 0.02s & 206.82s & 3.87s & 28.76s  \\
        \hline
    \end{tabular}
\vspace{1mm}
\caption{Time Measurements for Ping Pong Protocol}
\label{tab:run-time}
\vspace{-7mm}
\end{table}

\mypara{Type-checking Time of Generated Code and Callbacks}
We measure the time taken for the generated APIs to type-check in \fstar.
We provide a simple \fstar implementation of the callbacks following
the generated APIs, and measure the time taken to type-check the callbacks.

The increase of type-checking time is non-linear with regard to the protocol
length.
We encode CFSM states as records corresponding to local typing contexts.
In this case, the size of local typing contexts and the number of type
definitions grows linearly, giving rise to a non-linear increase.
Moreover, the entry point function
is likely to cause non-linear increases in the type-checking time.
The long type-checking time of the generated code could be avoided if the
developer chooses to trust our toolchain to always generate well-typed \fstar
code for the entry point.
The entry point would be available in an \emph{interface file} (cf.\ OCaml
\code{.mli} files), with the actual implementation in OCaml instead of
\fstar\footnotemark.
There would otherwise be no changes in the development workflow.
Although neither does type-checking time of the callback implementation
fit a linear pattern, it remains within reasonable time frame.
\footnotetext{Defining a signature in an interface file, and providing an
  implementation in the target language (OCaml) allows the \fstar
  compiler to \emph{assume} the implementation is correct.
  This technique is
  used frequently in the standard library of \fstar.
  This is not to be confused with implementing the endpoints in OCaml
  instead of \fstar,
  as that would bypass the \fstar type-checking.
}

\subsection{Runtime Performance (Execution Time)}
\label{subsection:eval-exec}

We measure the execution time taken for an exchange of 100,000 ping pongs for
the toolchain and bare implementation under the experiment network.
The execution time is dominated by network communication, since there is little
computation to be performed at each endpoint.

We provide a bare implementation using a series of direct invocations of
sending and receiving primitives, in a compatible way to communicate with
generated APIs.
The bare implementation does not involve a record of callbacks, which is
anticipated to run faster, since the bare implementation involves fewer
function pointers when calling callbacks.
Moreover, the bare implementation does not construct \emph{state records},
which record a backlog of the communication, as the protocol progresses.
To measure the performance impact of book-keeping of callback and state records,
we run the Ping Pong protocol from \cref{fig:pingpong} for
a protocol of increasing size (number of states and generated types),
i.e.\ for increasing values of  $n$. All implementations,
including \textit{bare} are run until 100,000 ping pong messages in total are
exchanged\footnote{For $n=1$, we run 100,000 iterations of recursion;
for $n=10$, we run 10,000 iterations, etc. Total number of ping pong messages
exchanged by two parties remain the same. }.

We summarise the results in \cref{tab:run-time}.
Despite the different protocol lengths, there are \emph{no significant changes}
in execution time.
Since the execution is dominated by time spent on communication,
the measurements are subject to network fluctuations, difficult to avoid
during the experiments.
We conclude that our implementation does not impose a large overhead on the
execution time.
\subsection{Expressiveness}
\label{subsection:eval-expressive}

We implement examples from the session type literature, and add
refinements to encode data dependencies in the protocols.
We measure the time taken for code generation and type-checking, and
present them in \cref{tab:examples}.
The time taken in the toolchain for examples in the session type literature
is usually short, yet we demonstrate that we are able to implement the examples
easily with our callback style API.
Moreover, the time taken is incurred at the compilation stage, hence there
is no overhead for checking refinements by our runtime.

\newcommand{\xmark}{\ding{55}}
\begin{table}
  \centering
  \small
  \begin{tabular}{|l|r|c|c|c|c|}
    \hline
    Example (Endpoint) & Gen.\@ / TC.\@ Time & MP & RV & IV & STP \\ 
    \hline
    \hline
    Two Buyer $^a$ (\code{A}) & 0.46s / 2.33s & \checkmark & \xmark & \checkmark & \checkmark
    ${}^\dagger$ \\
    \hline
    Negotiation $^b$ (\code{C})& 0.46s / 1.59s & \xmark & \checkmark & \xmark & \xmark \\
    \hline
    Fibonacci $^c$ (\code{A}) & 0.44s / 1.58s & \xmark & \checkmark & \xmark & \xmark \\
    \hline
    Travel Agency $^d$ (\code{C})& 0.62s / 2.36s & \checkmark & \xmark & \xmark & \checkmark${}^\dagger$\\
    \hline
    Calculator $^c$ (\code{C})& 0.51s / 2.30s & \xmark & \xmark & \xmark & \checkmark \\
    \hline
    SH $^e$ (\code{P}) & 1.16s / 4.31s & \checkmark & \xmark & \checkmark & \checkmark ${}^\dagger$\\
    \hline
    Online Wallet $^f$ (\code{C}) & 0.62s / 2.67s & \checkmark &
    \checkmark & \xmark & \xmark \\
    \hline
    Ticket $^g$ (\code{C}) & 0.45s / 1.90s & \xmark
    & \checkmark & \xmark & \xmark \\
    \hline
    HTTP $^h$ (\code{S}) & 0.55s / 1.79s & \xmark & \xmark &
    \xmark & \checkmark ${}^\dagger$\\
    \hline
  \end{tabular}%
  \footnotesize
  \begin{tabular}{r|l}
MP & Multiparty Protocol \\
RV & Uses Recursion Variables \\
IV & \begin{tabular}{@{}l@{}}Irrelevant Variables
\end{tabular}\\
STP &
\begin{tabular}{@{}l@{}}
Implementable in STP\\
\checkmark ${}^\dagger$ STP requires \emph{dynamic} checks
\end{tabular} \\
$^a$ & \cite{JACM16MPST} \\
$^b$ & \cite{DBLP:conf/concur/DemangeonH12}\\
$^c$ & \cite{FASE16EndpointAPI} \\
$^d$  & \cite{DBLP:conf/ecoop/HuYH08} \\
$^e$ & \cite{CC18TypeProvider} \\
$^f$ & \cite{RV13SPy} \\
$^g$ & \cite{TGC12MultipartyMultiSession} \\
$^h$ & \cite{RFCHttpSyntax}
  \end{tabular}
  \caption{Selected Examples from Literature}
  \label{tab:examples}
  \vspace{-10mm}
\end{table}

We also compare the expressiveness of our work with two most closely related works, namely \citet{CONCUR10DesignByContract} and \citet{CC18TypeProvider}, which study refinements in MPST (also see \cref{section:related}).
\Citet{CC18TypeProvider} (Session Type Provider, STP) implements
limited version of refinements in the \scrib toolchain.
Our version is strictly more expressive than STP for two reasons:
(1) support for recursive variables to express invariants
and (2) support for irrelevant variables.
\cref{fig:new-old-scribble} illustrates those features
and \cref{tab:examples} identifies
which of the implemented examples use them.
\input{fig/fig-new-old-scribble.tex}

In STP, when recursion occurs, all information about the
variables is lost at the end of an iteration, hence
their tool does not support even the simple example in
\cref{fig:new-scribble-adder}.
In contrast, our work retains the recursion variables,
which are available throughout the recursion.
Additionally, the endpoint projection in STP is more conservative with regards
to refinements. Whilst there must be no variables unknown to a role in the
refinements attached to a message for the sending role, there may be unknown
variables to the receiving role.
The part unknown to the receiving role is discarded (hence
weakening the pre-condition). In our work such information can still be
retained and used for type checking, thanks to irrelevant variables.

In \citet{CONCUR10DesignByContract}, a global protocol with assertions must be
\emph{well-asserted} (\S 3.1).
In particular, the \emph{history sensitivity} requirement states:
    \textit{"A predicate guaranteed by a participant $\ppt p$ can only contain those
    interaction variables that $\ppt p$ knows."}
Our theory lifts this restriction by allowing variables unknown to a
sending role to be used in the global or local type, whereas such variables
cannot be used in the implementation.
For example, \cref{example:gty} fails
the well-asserted requirement in \cite{CONCUR10DesignByContract}.
In the refinement $\dexp{x = z}$ for variable $\dexp z$ (for message label $\dlbl{Trd}$),
the variable $\dexp x$ is not known to $\ppt C$, hence the protocol would not
be well-asserted.
In our setup, such protocol is permitted, the endpoint implementation for $\ppt
C$ can provide the value $\dexp y$ received from $\ppt B$ to satisfy the
refinement type --- The SMT solver can validate the refinement from the
transitivity of equality.

%% file: fig/fig-new-old-scribble.tex
\begin{figure}[ht]
  \vspace{-4mm}
    \centering
    \begin{subfigure}[b]{0.48\textwidth}
        \centering
\begin{lstlisting}[language=Scribble, numbers=none]
protocol Adder(role S, role C)
@"S[acc:=0]" {
  Num(x:int)   from C to S; @"x>=0"
  Sum(sum:int) from S to C; @"sum=acc+x"
  do Adder(S, C); @"S[sum]" }
\end{lstlisting}
    \vspace{-3mm}
        \caption{Accumulator (using Recursive Invariants)}
        \label{fig:new-scribble-adder}
    \end{subfigure}
    \begin{subfigure}[b]{0.48\textwidth}
        \centering
\begin{lstlisting}[language=Scribble, numbers=none]
protocol Broadcast(role A, role B, role C) 
{
  Broadcast(x:int) from A to B; @"x>=0"
  // C does not learn y>=0 in STP
  Broadcast(y:int) from A to C; @"x=y" }
\end{lstlisting}
    \vspace{-3mm}
        \caption{Broadcasting (using Irrelevant Variables)}
        \label{fig:new-scribble-2}
    \end{subfigure}
    \vspace{-1mm}
    \caption{Example Protocols Demonstrating Additional Expressiveness to
      \cite{CC18TypeProvider}}
    \label{fig:new-old-scribble}
    \vspace{-3mm}
\end{figure}

%% file: related_work.tex
\section{Related Work}
\label{section:related}
We summarise the most closely related works in the areas of
refinement and session types.
For a detailed survey on theory and implementations of session types, see \citet{BettyBook}.

\mypara{Refinement Types for Verification and Reasoning}
\label{subsection:related-refinement}
Refinement types were introduced
to allow recursive data structures to be specified in more details
using predicates \cite{PLDI91RefinementML}.
Subsequent works on the topic \cite{TOPLAS11Refinement, ICFP14LiquidHaskell,
SCALA16Refinement, POPL18Refinement}  utilise SMT solvers,
 such as Z3 \cite{TACAS08Z3}, to aid the type system to
decide a semantic subtyping relation \cite{JFP12SemanticSubtyping} using
SMT encodings.
Refinement types have been applied to numerous domains, such as resource usage
analysis \cite{POPL20Refinement, ICFP20RefinementResource}, secure implementations
\cite{POPL10F7Protocol, TOPLAS11Refinement}, information control flow
enforcements \cite{ICFP20RefinementIFC}, and theorem proving
\cite{POPL18Refinement}.
Our aim is to utilise refinement types for the
specification and
verification of distributed protocols,
by combining refinement and session types
in a single practical framework.

\mypara{Implementation of Session Types}
\label{subsection:related-session-impl}
\Citet{CC18TypeProvider}
provides an implementation of MPST with assertions
using \scrib and \fsharp.
Their implementation, the session type provider (STP), relies on
code generation of fluent (class-based) APIs, initially described in \cite{FASE16EndpointAPI}.
Each protocol state is implemented as a class,
with methods corresponding to the possible transitions from that state.
It forces a programming style that not only relies extensively on method chaining,
but also requires dynamic checks to ensure the
linearity of channel usage.
Our work differs from STP in multiple ways.
First, we extend the \scrib toolchain to support \emph{recursion variables},
    allowing refinements on recursions, hence improving expressiveness.
    In this way, developers can specify dependencies across recursive calls,
    which is not supported in STP.
Second, we depart from the class-based API generation, and
    generate a callback-based API.
    Our approach has the advantage that the linear usage of channels is ensured
    by construction, saving dynamic checks for channels.
Third, we use refinement types in \fstar to verify refinements
    statically, in
    contrast, STP performs dynamic evaluations to validate assertions in protocols.
    Finally,
    the metatheory of session types extended with refinements was not developed
    in their work.

Several other MPST works follow a similar technique of class-based API generation to overcome
limitations of the type system in the target language, e.g.
\citet{POPL19Parametric} for Go, \citet{CC15MPI} for C.
All of the above works,
suffer from the same limitations --
they detect linearity violations at runtime, and offer no static alternative.
Indeed, to our knowledge,
\citet{ECOOP20OCamlMPST} provide the only MPST implementation
which \textit{statically} checks linearity violation.
It relies on specific type-level OCaml features,
and a monadic programming style.
Our work proposes generation of a callback-styled API from MPST protocols. To our knowledge,
it is the first work that ensures linear channel usage by construction.
Although our target language is  F*, the callback-styled API code generation technique
is applicable to any mainstream programming language.

\mypara{Dependent and Refinement Session Types}
\label{subsection:related-mpst}
\Citet{CONCUR10DesignByContract} propose a multiparty session
$\pi$-calculus with logical assertions.
By contrast, our formulation of RMPST is based on refinement types, projection
with silent prefixes and correspondence with CFSMs, to target practical code
generation, such as for \fstar.
They do not formulate any semantics for global types nor prove an equivalence
between refined global types and projections, as in this paper.
\Citet{JLAMP17DependentMPST} extend MPST with value dependent types.
Invariants on values are witnessed by proof objects, which then may be erased
at runtime.
Our work uses refinement types, which follows the principle naturally, since
refinements that appear in types are proof-irrelevant and can be
erased safely.
These works are limited to theory, whereas we provide an implementation.

\Citet{PLACES19DependentSession} propose an Embedded
Domain Specific Language (EDSL) approach of implementing multiparty sessions
(analogous to MPST) in Idris.
They use value dependent types in Idris to define combinators,
with options to specify data dependencies, contrary to our approach of code
generation.
However, the combinators only describe the sessions, and how to implement and
execute the sessions remains unanswered.
Our work provides a complete toolchain from protocol description to
implementation and verification.

In the setting of binary session types,
\Citet{CONCUR20SessionRefinement} extend session types with arithmetic
refinements, with application to work analysis for computing upper bounds of
work from a given session type.
\Citet{POPL20LabelDependent} extend binary session types with label dependent
types.
In the setup of their work, specification of arithmetic properties involves
complicated definitions of inductive arithmetic relations and functions.
In contrast, we use SMT solvers, which have built-in functions and relations
for arithmetic.
Furthermore, there is no need to construct proofs manually, since SMT solvers
find the proof automatically, which enhances usability and ergonomics.
\Citet{POPL20Actris} combine binary session types with concurrent separation
logic, allowing reasoning about mixed-paradigm concurrent programs, and planned
to extend the framework to MPST.
Along similar lines, \citet{ICFP20SteelCore} provide a framework of concurrent
separation logic in \fstar, and demonstrate its expressiveness by showing how
(dependent) binary session types can be represented in the logic and used in
reasoning.
Our work is based on the theory of MPST, subsuming the binary session types.
Furthermore, we implement a toolchain that developers can use.

\Citet{CSF09Session} use refinement types to implement a limited form of
multiparty session types.
Session types are encoded in refinement types via code generation.
The specification language they use, albeit similar to MPST,
has limited expressive power.
Only patterns of interactions where participants
alternate between sending and receiving are permitted.
Moreover, they do not study data dependencies in protocols,
hence they can neither specify, nor verify
constraints on payloads or recursions.
We use refinement types to specify constraints and
dependencies in multiparty protocols, and use the \fstar
compiler \cite{POPL16FStar} for verifying the
endpoint implementations.
The verified endpoint program does not only comply to the multiparty
protocol, enjoying the guarantees provided by the original MPST theory (deadlock
freedom, session fidelity), but also satisfies additional guarantees provided
by refinement types with respect to data constraints.

%% file: conclusion.tex
\section{Conclusions and Future Work}
\label{section:conclusion}

We present a novel toolchain for implementing refined multiparty session types
(RMPST), which enables developers to use \scrib, a protocol description
language for multiparty session types, and \fstar, a state-of-the-art
verification-oriented programming language, to implement a
multiparty protocol and statically verify endpoint implementations.
To the best of the authors' knowledge, this is the first work on
\emph{statically} verified multiparty protocols with \emph{refinement} types.
We extend the theory of multiparty session types
with data refinements,
and present a toolchain that enables developers to \emph{specify} multiparty
protocols with data dependencies, and \emph{implement} the endpoints using
generated APIs in \fstar.
We leverage the advanced typing system in \fstar to encode local session types
for endpoints, and validate the data dependencies in the protocol statically.

The verified endpoint program in \fstar is extracted into
OCaml, where the refinements are \emph{erased} ---
adding \emph{no runtime overhead} for refinements.
The callback-styled API avoids linearity checks of channel usage by
internalising communications in generated code.
We evaluate our toolchain and demonstrate that our overhead is small
compared to an implementation without session types.

Whereas refinement types express the data dependencies of multiparty protocols,
the availability of refinement types in general purpose mainstream programming
languages is limited.
For future work, we wish to study how to mix participants with refined
implementation and those without, possibly using a gradual typing
system \cite{POPL17GradualRefinement, JFP19GradualSession}.


%% file: acks.tex
\begin{acks}                            
  We thank the OOPSLA reviewers for their comments and suggestions,
  and David Castro, Julia Gabet and Lorenzo Gheri for proofreading.
  We thank Wei Jiang and Qianyi Shu for testing the artifact.
  The work is supported by EPSRC EP/T006544/1, EP/K011715/1, EP/K034413/1,
  EP/L00058X/1, EP/N027833/1, EP/N028201/1, EP/T006544/1, EP/T014709/1 and
  EP/V000462/1, and NCSS/EPSRC VeTSS.
\end{acks}

%% file: proof.tex
\subsection{Auxiliary lemmas}
\begin{lemma}
  Given a participant $\ppt p$,
  a global typing context $\Gamma$ and
  a local typing context $\Sigma$
  such that $\ctxproj{\Gamma}{p} = \Sigma$.

  Then, the projection of global typing context $\Gamma$ with
  $\vb{x^{\bigP}}{T}$ to $\ppt p$ satisfies that
  $$\ctxproj{\ctxext{\Gamma}{x^{\bigP}}{T}}{p} = \begin{cases}
    \ctxext{\Sigma}{x^\omega}{T} & \text{if~} \ppt p \in \bigP \\
    \ctxext{\Sigma}{x^0}{T} & \text{if}~ \ppt p \notin \bigP
  \end{cases}
  $$
  \label{lem:project-ext}
\end{lemma}
\begin{proof}
  By expanding defintion and case analysis.
\end{proof}

\begin{lemma}
  Given a participant $\ppt p$,
  global typing contexts $\Gamma_1, \Gamma_2$,
  a global type $\dgt G$.

  If the two typing contexts have the same projection on $\ppt p$:
  $\ctxproj{\Gamma_1}{p} = \ctxproj{\Gamma_2}{p}$, then the projection of
  global types under the two contexts are the same:
  $\gtctxproj{\Gamma_1}{G}{p} = \gtctxproj{\Gamma_2}{G}{p}$.
  \label{lem:project-env}
\end{lemma}
\begin{proof}
  By induction on the derivation of projection of global types.
\end{proof}

\begin{lemma}
  Given a global typing context $\Gamma$, a set of participants $\bigP$, a
  participant $\ppt p$, a variable $\dexp x$, a well-formed type $\dte T$ and a
  global type $\dgt G$.

  If $\gtctxproj{\ctxext{\Gamma}{x^{\varnothing}}{T}}{G}{p} =
  \ltctx{\Sigma}{L}$, then $\gtctxproj{\ctxext{\Gamma}{x^{\bigP}}{T}}{G}{p} =
  \ltctx{\Sigma}{L}$.

  \label{lem:project-weaken}
\end{lemma}

\begin{proof}
  By induction on derivation of projection of global types, via
  weakening of local typing rules: $$\Sigma_1, \vb{x^0}{T}, \Sigma_2 \vdash
  \vb{E}{T'} \implies \Sigma_1, \vb{x^\omega}{T}, \Sigma_2 \vdash \vb{E}{T'}$$
\end{proof}

\begin{lemma}[Inversion of Projection]
  Given a global typing context $\Gamma$, a global type $\dgt G$, a participant
  $\ppt p$, a local typing context $\Sigma$, a local type $\dtp L$, such that
  $\gtctxproj{\Gamma}{G}{p} = \ltctx{\Sigma}{L}$.

  If $\dtp L$ is of form:
  \begin{enumerate}
    \item $\dtp L = \ttake{q}{\dlbl{l_i}(\vb{x_i}{T_i}).L_i}_{i \in I}$,
      then $\dgt G$ is of form $\gtbran{p}{q}{\dlbl{l_i}(\vb{x_i}{T_i}).G_i}_{i
      \in I}$.
    \item $\dtp L = \toffer{q}{\dlbl{l_i}(\vb{x_i}{T_i}).L_i}_{i \in I}$,
      then $\dgt G$ is of form $\gtbran{q}{p}{\dlbl{l_i}(\vb{x_i}{T_i}).G_i}_{i
      \in I}$.
    \item $\dtp L = \tphi{l}{x}{T}{L}$, then $\dgt G$ is of form
      $\gtbran{s}{t}{\dlbl{l_i}(\vb{x_i}{T_i}).G_i}_{i \in I}$, where $\ppt p
      \notin \psetof{s, t}$ and $\tphi{l}{x}{T}{L} = \sqcup_{i \in
        I}{\tphi{l_i}{x_i}{T_i}{L_i}}$, where $\dtp{L_i}$ is obtained via the
      projection of $\dgt{G_i}$.
    \item $\dtp L = \trecur{t}{}{\vb{x}{T}}{\dexp{x := E}}{L'}$,
      then $\dgt G$ is of form $\gtrecur{t}{\vb{x}{T}}{\dexp{x := E}}{G'}$,
      and $\ppt p \in \dgt G$.
  \end{enumerate}
  \label{lem:proj-inv}
\end{lemma}

\begin{proof}
  A direct consequence of projections rules $\gtctxproj{\Gamma}{G}{p}$ ---
  The results of projections are not overlapping.
\end{proof}

\begin{lemma}[Determinancy]
  Let $\gtctx{\Gamma}{G}$ be a global type under a global typing context, and
  $\alpha$ be a labelled action.

  If $\gtctx{\Gamma}{G} \stepsto[\alpha] \gtctx{\Gamma_1}{G_1}$ and
  $\gtctx{\Gamma}{G} \stepsto[\alpha] \gtctx{\Gamma_2}{G_2}$,
  then $\gtctx{\Gamma_1}{G_1} = \gtctx{\Gamma_2}{G_2}$.
  \label{lem:gty-determinancy}
\end{lemma}

\begin{proof}
  By induction on global type reduction rules.
\end{proof}

\subsection{Proof of \cref{thm:trace-eq}}
\begin{quote}
  Let $\gtctx{\Gamma}{G}$ be a global type under a global typing context and $s
  \ltassoc \gtctx{\Gamma}{G}$
  be a configuration associated with the global type and context.

  $\gtctx{\Gamma}{G} \stepsto[\alpha] \gtctx{\Gamma'}{G'}$ if and only if $s
  \stepsto[\alpha] s'$, where $s' \ltassoc \gtctx{\Gamma'}{G'}$.
\end{quote}
\begin{proof}
  Soundness ($\Rarr$):
  \input{soundness.tex}

  Completeness ($\Larr$):
  \input{completeness.tex}

\end{proof}

\subsection{Proof of \cref{thm:wf-preservation}}
\begin{quote}
  If $\gtctx{\Gamma}{G}$ is a well-formed global type under typing context, and
  $\gtctx{\Gamma}{G} \stepsto[\alpha] \gtctx{\Gamma'}{G'}$, then
  $\gtctx{\Gamma'}{G'}$ is well-formed.
\end{quote}
\begin{proof}
  By induction on the reduction of global type $\gtctx{\Gamma}{G}
  \stepsto[\alpha]\gtctx{\Gamma'}{G}$.

  \begin{itemize}
    \item \ruleGPfx~
      $\gtctx{\Gamma}{\gtbran{p}{q}{\dlbl{l_i}(\vb{x_i}{T_i}).G_i}_{i \in I}}
      \stepsto[\ltsmsg{p}{q}{l_j}{x_j}{T_j}]
      \gtctx{\ctxext{\Gamma}{x_{j}^{\psetof{p, q}}}{T_j}}{\dgt{G_j}}$

      There are three cases for projection to consider: \rulePSend, \rulePRecv,
      and \rulePPhi.
      In all cases, the premises state that
        $\gtctxproj{\ctxext{\Gamma}{x_{i}^{\psetof{p, q}}}{T_i}}{G_i}{p} =
        \ltctx{\Sigma_i}{L_i}$,
      which indicates that all continuations are projectable for all indices.

    \item \ruleGCtx~
      $\gtctx{\Gamma}{\gtbran{p}{q}{\dlbl{l_i}(\vb{x_i}{T_i}).G_i}_{i \in I}}
      \stepsto[\alpha]
      \gtctx{\Gamma'}{\gtbran{p}{q}{\dlbl{l_i}(\vb{x_i}{T_i}).G_i'}_{i \in I}}$

      From inductive hypothesis, we have that for all index $i \in I$, if
      $\gtctx{\ctxext{\Gamma}{x_{j}^{\varnothing}}{T_j}}{G_j}$ is well-formed,
      then $\gtctx{\Gamma'}{G_j'}$ is well-formed.

      In all three cases of projection, the premises state that
        $\Sigma \vdash \dte{T_i} ~ty$ for all index $i \in I$.
      Therefore, the context
      $\gtctx{\ctxext{\Gamma}{x_{j}^{\varnothing}}{T_j}}{G_j}$ is also
      well-formed.

      We are left to show that
      $\gtctx{\Gamma'}{\gtbran{p}{q}{\dlbl{l_i}(\vb{x_i}{T_i}).G_i'}_{i \in I}}$
      is well-formed: we invert the premise of projection of
      $\gtctx{\Gamma}{\gtbran{p}{q}{\dlbl{l_i}(\vb{x_i}{T_i}).G_i}_{i \in I}}$
      and apply \cref{lem:project-weaken}.

    \item \ruleGRec~
      $\gtctx{\Gamma}{\gtrecur{t}{\vb{x}{T}}{\dexp{ x := E}}{G}}
      \stepsto[\alpha]
      \gtctx{\Gamma'}{G'}$

      By inductive hypothesis.
  \end{itemize}
\end{proof}

\subsection{Proof of \cref{thm:progress}}
\begin{quote}
  If $\gtctx{\Gamma}{G}$ is a well-formed global type under context, then $\gtctx{\Gamma}{G}$
  satisfies progress.
\end{quote}
\begin{proof}
  By induction on the structure of global types $\dgt G$.
  Using \cref{thm:trace-eq}, we are sufficient to show that
  $\gtctx{\Gamma}{G} \stepsto[\alpha] \gtctx{\Gamma'}{G'}$, and apply the
  theorem for the progress of associated configuration.

  \begin{itemize}
    \item $\dgt G = \gtbran{p}{q}{\dlbl{l_i}(\vb{x_i}{T_i}). G_i}_{i \in I}$.

    Since the index set $I$ must not be empty, we can pick an index $i \in I$
    and apply \ruleGPfx.

    \item $\dgt G = \gtrecur{t}{\vb{x}{T}}{\dexp{x := E}}{G'} $

    Since recursive types must be contractive, we have that
    $\dgt G'\subst{\gtrecursimpl{t}{\vb{x}{T}}{G'}}{\dgt{\mathbf{t}}} \neq
    \dgt{G'}$.
    Furthermore, the substituted type is closed. We can apply \ruleGRec.

    \item $\dgt G = \gtvar{t}{\dexp{x := E}}$

    Vacuous, since well-formed global type cannot have free type variable.

    \item $\dgt G = \gtend$

    Corresponds to the case where all local types are $\tend$.
  \end{itemize}

\end{proof}

%% file: soundness.tex
By induction on the reduction rules of global types $\gtctx{\Gamma}{G}
\stepsto[\alpha] \gtctx{\Gamma'}{G'}$.

\begin{enumerate}
  \item \ruleGPfx

    For \ppt p:

    We project the global context before transition to $\ppt p$.

    $ \gtctxproj{\Gamma}{\gtbran{p}{q}{\dlbl{l_i}(\vb{x_i}{T_i}).G_i}_{i \in I}}{p}
     = \ltctx{\Sigmap}{\ttake{q}{\dlbl{l_i}(\vb{x_i}{T_i}).L_i}_{i \in I}} $

    where
    $\Sigmap = \ctxproj{\Gamma}{p}$
    and
    $\ltctx{\Sigma_i}{L_i} = \gtctxproj{\ctxext{\Gamma}{x_{i}^{\psetof{p, q}}}{T_i}}{G_i}{p}$
    for $i \in I$

    We have
    $\ltctx{\Sigmap}{\ttake{q}{\dlbl{l_i}(\vb{x_i}{T_i}).L_i}_{i \in I}}
    \stepsto[\ltsmsg{p}{q}{l_j}{x_j}{T_j}]
    \ltctx{\ctxext{\Sigmap}{x_j}{T_j}}{L_j}$
    by \ruleLSend.

    We project the global context after transition to $\ppt p$ (via \cref{lem:project-ext}).

    $
      \gtctxproj{\ctxext{\Gamma}{x_{j}^{\psetof{p, q}}}{T_j}}{\dgt{G_j}}{p}
      = \ltctx{\ctxext{\Sigmap}{x_j}{T_j}}{L_j}
    $

    which is the same as the result of applying \ruleLSend.

    For $\ppt q$:

    Similar to the case of $\ppt p$, using \ruleLRecv.

    For $\ppt r$ ($\ppt r \neq \ppt p$, $\ppt r \neq \ppt q$):

    We project the global context before transition to $\ppt r$.

    $ \gtctxproj{\Gamma}{\gtbran{p}{q}{\dlbl{l_i}(\vb{x_i}{T_i}).G_i}_{i \in I}}{p}
     = \ltctx{\Sigmar}{\sqcup_{i \in I} \tphi{l_i}{x_i}{T_i}{L_i}}$

    where
    $\Sigmar = \ctxproj{\Gamma}{r}$
    and
    $\ltctx{\Sigma_i}{L_i} = \gtctxproj{\ctxext{\Gamma}{x_{i}^{\psetof{p, q}}}{T_i}}{G_i}{r}$
    for $i \in I$

    We project the global context after transition to $\ppt r$ (via \cref{lem:project-ext}).

    $
      \gtctxproj{\ctxext{\Gamma}{x_{j}^{\psetof{p, q}}}{T_j}}{\dgt{G_j}}{r}
      = \ltctx{\ctxext{\Sigmar}{x_j^0}{T_j}}{L_j}
    $

    We have
    $
      \ltctx{\Sigmar}{\sqcup_{i \in I} \tphi{l_i}{x_i}{T_i}{L_i}}
      \ltequiv
      \ltctx{\ctxext{\Sigmar}{x_j^0}{T_j}}{L_j}
    $ by \ruleEPhi~and plain merging.
  \item \ruleGCtx

    By inductive hypothesis, we have

    \begin{quote}
      $\forall j \in I$.
      $\setof{\gtctxproj{\ctxext{\Gamma}{x_{j}^{\varnothing}}{T_j}}{G_j}{r}}_{\ppt r \in \dgt{G_j}}
      \stepsto[\alpha]
      \setof{\gtctxproj{\Gamma'}{G_j'}{r}}_{\ppt r \in \dgt{G_j}}$

      If $\ppt r \in \subj{\alpha}$,
      then
      $\gtctxproj{\ctxext{\Gamma}{x_{j}^{\varnothing}}{T_j}}{G_j}{r}
      \stepsto[\alpha]
      \gtctxproj{\Gamma'}{G_j'}{r}$.

      If $\ppt {r'} \notin \subj{\alpha}$,
      then
      $\gtctxproj{\ctxext{\Gamma}{x_{j}^{\varnothing}}{T_j}}{G_j}{r'}
      \ltequiv
      \gtctxproj{\Gamma'}{G_j'}{r'}$.
    \end{quote}

    For $\ppt p$:

    We project the global context before transition to $\ppt p$.

    $\gtctxproj{\Gamma}{\gtbran{p}{q}{\dlbl{l_i}(\vb{x_i}{T_i}).G_i}_{i \in I}}{p}
    = \ltctx{\Sigmap}{\ttake{q}{\dlbl{l_i}(\vb{x_i}{T_i}).L_i}}_{i \in I}$

    where
    $\Sigmap = \ctxproj{\Gamma}{p}$
    and
    $\ltctx{\Sigma_i}{L_i} = \gtctxproj{\ctxext{\Gamma}{x_{i}^{\psetof{p, q}}}{T_i}}{G_i}{p}$
    for $i \in I$

    We project the global context after transition to $\ppt p$.

    $\gtctxproj{\Gamma'}{\gtbran{p}{q}{\dlbl{l_i}(\vb{x_i}{T_i}).G_i'}_{i \in I}}{p}
    = \ltctx{\Sigmap'}{\ttake{q}{\dlbl{l_i}(\vb{x_i}{T_i}).L_i'}}_{i \in I}$

    where
    $\Sigmap' = \ctxproj{\Gamma'}{p}$
    and
    $\ltctx{\Sigma'_i}{L_i'} = \gtctxproj{\ctxext{\Gamma'}{x_{i}^{\psetof{p, q}}}{T_i}}{G_i'}{p}$
    for $i \in I$

    Since $\ppt p \notin \subj{\alpha}$,
    we have
    $\ltctx{\ctxext{\Sigmap}{x_j^0}{T_j}}{L_j}
    \ltequiv
    \ltctx{\Sigmap'}{L_j'} = \gtctxproj{\Gamma'}{G_j'}{p}$
    from the inductive hypothesis.

    By \ruleECtx, we have
    $\ltctx{\Sigmap}{\ttake{q}{\dlbl{l_i}(\vb{x_i}{T_i}). L_i}_{i \in I}}
    \ltequiv
    \ltctx{\Sigmap'}{\ttake{q}{\dlbl{l_i}(\vb{x_i}{T_i}). L_i'}_{i \in I}}$.

    For $\ppt q$, the proof is similar.

    For $\ppt r$, where $\ppt r \neq \ppt p$, $\ppt r \neq \ppt q$, and $\ppt r
    \notin \subj{\alpha}$:

    We project the global context before transition to $\ppt r$.

    $\gtctxproj{\Gamma}{\gtbran{p}{q}{\dlbl{l_i}(\vb{x_i}{T_i}).G_i}_{i \in I}}{r}
    = \ltctx{\Sigmar}{\sqcup_{i \in I}\tphi{l_i}{x_i}{T_i}{L_i}}$

    where
    $\Sigmar = \ctxproj{\Gamma}{r}$
    and
    $\ltctx{\Sigma_i}{L_i} = \gtctxproj{\ctxext{\Gamma}{x_{i}^{\psetof{p, q}}}{T_i}}{G_i}{r}$
    for $i \in I$

    We project the global context after transition to $\ppt r$.

    $\gtctxproj{\Gamma'}{\gtbran{p}{q}{\dlbl{l_i}(\vb{x_i}{T_i}).G_i'}_{i \in I}}{r}
    = \ltctx{\Sigmar'}{\sqcup_{i \in I}\tphi{l_i}{x_i}{T_i}{L_i'}}$

    where
    $\Sigmar' = \ctxproj{\Gamma'}{r}$
    and
    $\ltctx{\Sigma_i'}{L_i'} = \gtctxproj{\ctxext{\Gamma'}{x_{i}^{\psetof{p, q}}}{T_i}}{G_i'}{r}$
    for $i \in I$

    Since $\ppt r \notin \subj{\alpha}$, we have
    $\ltctx{\ctxext{\Sigmar}{x_j^0}{T_j}}{L_j}
    \ltequiv
    \ltctx{\Sigmar'}{L_j'} = \gtctxproj{\Gamma'}{G_j'}{r}$
    from the inductive hypothesis.

    By \ruleEPhi, we have
    $\ltctx{\Sigmar}{\sqcup_{i \in I}\tphi{l_i}{x_i}{T_i}{L_i}}
    \ltequiv
    \ltctx{\ctxext{\Sigmar}{x_j^0}{T_j}}{L_j}$ where $j \in I$

    By transitivity, we have
    $\ltctx{\Sigmar}{\sqcup_{i \in I}\tphi{l_i}{x_i}{T_i}{L_i}}
    \ltequiv
    \ltctx{\Sigmar'}{L_j'}$

    For $\ppt r$, where $\ppt r \neq \ppt p$, $\ppt r \neq \ppt q$, and $\ppt r
    \in \subj{\alpha}$:

    Since $\ppt r \in \subj{\alpha}$, we have
    $\ltctx{\ctxext{\Sigmar}{x_j^0}{T_j}}{L_j}
    \stepsto[\alpha]
    \ltctx{\Sigmar'}{L_j'} = \gtctxproj{\Gamma'}{G_j'}{r}$
    from the inductive hypothesis.

    By \ruleEPhi, we have
    $\ltctx{\Sigmar}{\sqcup_{i \in I}\tphi{l_i}{x_i}{T_i}{L_i}}
    \ltequiv
    \ltctx{\ctxext{\Sigmar}{x_j^0}{T_j}}{L_j}$ where $j \in I$

    By \ruleLEps, we have
    $\ltctx{\Sigmar}{\sqcup_{i \in I}\tphi{l_i}{x_i}{T_i}{L_i}}
    \stepsto[\alpha]
    \ltctx{\Sigmar'}{L_j'}$
  \item \ruleGRec

    By inductive hypothesis.
\end{enumerate}

%% file: completeness.tex
  Fix the transition label $\alpha = \ltsmsg{p}{q}{l}{x}{T}$.

  We prove by induction on the reduction of $s \stepsto[\alpha] s'$, and
  consider the local types for after projection.

  \begin{itemize}
    \item
      \ruleLSend~and \ruleLRecv.
      This case arises from the projection of
      $\dgt G = \gtbran{p}{q}{\dlbl{l_i}(\vb{x_i}{T_i}).G_i}_{i \in I}$
      (\cref{lem:proj-inv} case (1) and (2)).

      Let $\Gamma' = \ctxext{\Gamma}{x_{\psetof{p, q}}}{T}$.

      By projection (\rulePSend, \rulePRecv), we have
      $\ltctx{\Sigmap}{\Lp} = \ttake{q}{\dlbl{l_i}(\vb{x_i}{T_i}).\Lp_i}$
      and \linebreak
      $\ltctx{\Sigmaq}{\Lq} = \toffer{p}{\dlbl{l_i}(\vb{x_i}{T_i}).\Lq_i}$,
      where $\Lp_i$ is obtained by $\gtctxproj{\Gamma'}{G_i}{p}$, resp.\ for
      $\ppt q$.

      Since $s \stepsto[\alpha] s'$,
      by inversion of \ruleLSend~on $\ltctx{\Sigmap}{\Lp}$ and \ruleLRecv~on
      $\ltctx{\Sigmaq}{\Lq}$, we have $j \in I$ with $\dlbl{l_j} = \dlbl{l}$,
      $\dexp{x_j} = \dexp{x}$ and $\dte{T_j} = \dte{T}$.
      By \ruleLSend~and \ruleLRecv, the configuration $s'$ has
      $\ltctx{\ctxext{\Sigmap}{x}{T}}{\Lp_j}$, resp.\ for $\ppt q$.

      We can obtain that
      $\gtctx{\Gamma}{G}
       \stepsto[l]
       \gtctx{\ctxext{\Gamma}{x_{j}^{\psetof{p, q}}}{T_j}}{G_j}$, via \ruleGPfx.

      For $\ppt p$ and $\ppt q$, the association with $s'$ is straight forward.
      We further show the association for $\ppt r$ ($\ppt r \neq \ppt {p, q}$).
      By projection (\rulePPhi) and plain merging, we have
      $\ltctx{\Sigmar}{\Lr} = \tphi{l}{x}{T}{\Lr_j}$, where $\Lr_j$ is obtained
      by the projection $\gtctxproj{\Gamma'}{G_j}{r}$.

      \ruleEPhi~can occur when $s \stepsto[l] s'$, and
      $\ltctx{\ctxext{\Sigmar}{x^0}{T}}{\Lr_j}$ is associated with $\gtctx{\Gamma}{G_j}$.
    \item
      \ruleLEps~(\ruleEPhi)~and \ruleLEps~(\ruleEPhi).
      This case arises from the projection of \linebreak
      $\dgt G = \gtbran{s}{t}{\dlbl{l_i}(\vb{x_i}{T_i}).G_i}_{i \in I}$
      ($\psetof{s, t} \cap \psetof{p, q} = \varnothing$, \cref{lem:proj-inv}
      case (3)).

      Let $\Gamma'_i = \ctxext{\Gamma}{x_{i}^{\psetof{s, t}}}{T}$,
      $\Gamma_{0i} = \ctxext{\Gamma}{x_{i}^{\varnothing}}{T}$.

      By projection (\rulePSend, \rulePRecv), we have
      $\ltctx{\Sigmap}{\Ls} = \ttake{q}{\dlbl{l_i}(\vb{x_i}{T_i}).\Ls_i}$
      and
      $\ltctx{\Sigmaq}{\Lt} = \toffer{p}{\dlbl{l_i}(\vb{x_i}{T_i}).\Lt_i}$,
      where $\Ls_i$ is obtained by $\gtctxproj{\Gamma'_i}{G_i}{s}$, resp.\ for
      $\ppt t$.

      For $\ppt r \neq \ppt {s, t}$, the projection gives a
      uniform silent prefix (\rulePPhi), $\ltctx{\Sigmar}{\Lr} = \sqcup_{i \in
      I}\tphi{l_i}{x_i}{T_i}{\Lr_i}$, where $\Lr_i$ is obtained by projection
      $\gtctxproj{\Gamma'}{G_i}{r}$.

      Let the merged label be $\dlbl{l_0}$, variable be $\dexp{x_0}$, value
      type be $\dte{T_0}$ and session type be $\dtp{\Lr_0}$.
      We have $\ltctx{\Sigmar}{\Lr} = \tphi{l_0}{x_0}{T_0}{\Lr_0}$.

      Let $\Gamma_{0} = \ctxext{\Gamma}{x_{0i}^{\varnothing}}{T}$.

      In this case, $\ruleLSend$~(resp.\ \ruleLRecv) cannot apply directly on
      $\Lp$ (resp.\ $\Lq$), so the rule applied must be \ruleLEps.

      By inversion of $\ruleLEps$ with $\ruleEPhi$, we have that
      $\ltctx{\ctxext{\Sigmap_i}{x_0^0}{T_0}}{\Lp}
      \stepsto[\alpha]
      \ltctx{\Sigmap'}{\Lp'}$ (resp.\ $\ppt q$).

      Take arbitrary $i \in I$, let $s_i$ be the configuration associated to
      $\gtctx{\Gamma_0}{G_i}$.

      We note that
        $\ltctx{\Sigmar}{\Lr} = \tphi{l_0}{x_0}{T_0}{\Lr_0}
         \stepsto[\epsilon]
         \ltctx{\ctxext{\Sigmar}{x_0^0}{T_0}}{\Lr_0} = \ltctx{\Sigmar_i}{\Lr_i}$ (\ruleEPhi).

      Despite that $\Lr_i$ is obtained from projection with context $\Gamma'$,
      the same local type is projected with context $\Gamma_0$
      (\cref{lem:project-env}).
      Therefore, we can apply the action
      $\alpha$ on $s_i$, and use the inductive hypothesis on
      $\gtctx{\Gamma_0}{G_i}$ and $s_i$.

      We thus have with arbitrary $i \in I$, $\gtctx{\Gamma_0}{G_i} \stepsto[\alpha]
      \gtctx{\Gamma_{0i}'}{G_i'}$ for some $\gtctx{\Gamma_{0i}'}{G_i'}$
      associated with $s_i'$.
      By \cref{lem:gty-determinancy}, the $\gtctx{\Gamma_{0i}'}{G_i'}$ are
      identical, we denote it $\gtctx{\Gamma_0'}{G'}$.

      We can therefore obtain that
      $\gtctx{\Gamma}{G}
       \stepsto[\alpha]
       \gtctx{\Gamma_{0}'}{\gtbran{s}{t}{\dlbl{l_i}(\vb{x_i}{T_i}).G'}_{i \in I}}$
      via \ruleGCtx.

      We are now left to show the association with $s'$:

      From the inductive hypothesis, we have $s_i' \ltassoc
      \gtctx{\Gamma_0'}{G'}$.

      For $\ppt r \neq \ppt {s, t}$, the association result follows from
      \cref{lem:project-env} and inductive hypothesis, since $\Gamma'_i$
      projects $\dexp{x_i}$ to an irrelevant quantified variable.
      For $\ppt {s, t}$, we use \cref{lem:project-weaken}, which weakens the
      projected $\dexp{x_i}$ to a concrete variable.

    \item
      \ruleLEps~(\ruleERec)~and \ruleLEps~(\ruleERec).
      This case arises from the projection of \linebreak
      $\dgt G = \gtrecur{t}{\vb{x'}{T'}}{\dexp{x' := E}}{G_1}$
      (\cref{lem:proj-inv}, case (4)).

      We discuss the case where the projection uses \rulePRecOne, since
      \rulePRecTwo~projects to $\tend$ with no actions to be taken.

      Let $\Gamma' = \ctxext{\Gamma}{x'^{\setof{\ppt r | \ppt r \in \dgt G_1}}}{T}$.

      For all roles $\ppt r \in \dgt G_1$, by projection (\rulePRecOne), we have
      $\ltctx{\Sigmap}{\trecur{t}{}{\vb{x'}{T'}}{\dexp{x' := E}}{\Lr_1}}$, where
      where $\Lr_1$ is obtained by $\gtctxproj{\Gamma'}{G_1}{r}$.

      In this case, \ruleLSend~(resp.\ \ruleLRecv) cannot apply directly on
      $\Lp$ (resp.\ $\Lq$), so the rule applied must be \ruleLEps.

      By inversion
      of \ruleLEps~with \ruleERec, we have that
      $\ltctx{\ctxext{\Sigmap}{x'^\omega}{T'}}{\Lp_1\subst{\trecursimpl{t}{}{\vb{x'}{T'}}{\Lp_1}}{\mathbf{t}}}
       \stepsto[\alpha]
       \ltctx{\Sigmap_1'}{\Lp_1'}$.

      For other roles $\ppt s \neq \ppt {p, q}$, by inverting \ruleERec, we
      have that

      $\ltctx{\ctxext{\Sigmas}{x'^\omega}{T'}}{\Ls_1\subst{\trecursimpl{t}{}{\vb{x'}{T'}}{\Ls_1}}{\mathbf{t}}}
       \stepsto[\alpha]
       \ltctx{\Sigmas_1'}{\Ls_1'}$.

      Combining all roles, we can use the inductive hypothesis and obtain

      $ \gtctx{\ctxext{\Gamma}{x'^{\setof{\ppt r | \ppt r\in \dgt
              G_1}}}{T'}}{G_1\subst{\gtrecursimpl{t}{\vb{x'}{T'}}{G_1}}{\mathbf{t}}}
        \stepsto[\alpha]
        \gtctx{\Gamma'}{G'}$ via \ruleGRec, and association with
        $\gtctx{\Gamma'}{G'}$.
    \item
      The rest of the cases are vacuous, since inverting the projections
      (\cref{lem:proj-inv}) of
      $\Lp$ and $\Lq$ leads to incompatible shapes of global type $\dgt G$.

  \end{itemize}

%% file: implementation_appendix.tex
\section{Additional Details on Code Generation}
\label{section:impl-appendix}

\subsection{Communicating Finite State Machines (CFSMs) (Toolchain Internals)}
\label{subsection:impl-cfsm}

\emph{Communicating Finite State Machines} (CFSMs, \cite{JACM83CFSM})
correspond to local types projected from global types, as shown in
\cite{ICALP13CFSM}.
We define the CFSM as a tuple $(\Q, q_0, \delta)$,
where $\Q$ is set of states, $q_0 \in \Q$ is an initial state,
and $\delta \subseteq \Q \times A \times \Q$ is a transition relation,
where $A$ is the set of labelled actions (cf.\ \cref{subsection:theory-semantics}).

\mypara{Conversion to Communicating Finite State Machines (CFSMs)}
\scrib follows the projection defined in \cref{subsection:theory-projection},
and projects a global protocol into local types.
Local types can be converted easily into a Communicating Finite State Machine
(CFSM), such that the resulting CFSM does not have mixed state (i.e.\ a state
does not contain a mixture of sending and receiving outgoing transitions), and
that the states are directed (i.e.\ they only contain sending or receiving
actions towards an identical participant)~\cite[Def. 3.4, Prop. 3.1]{ICALP13CFSM}.
We follow the same approach to obtain a CFSM from the local types with their
typing contexts.
The CFSM has the same trace of actions as the local types
\cite[Prop. 3.2]{ICALP13CFSM}.

We generate \fstar code from the CFSM obtained from projection.
We generate records for each state to correspond to the typing context
(explained in \cref{subsection:impl-typedef-st}), and functions for
transitions to correspond to actions (explained in
\cref{subsection:impl-handler}).
The execution of the CFSM is detailed in \cref{subsection:impl-run-fsm}.

\subsection{Generated APIs with Refinement Types (Toolchain Output)}
\label{subsection:impl-api}
Our code generator takes a CFSM as an input to produce type definitions and an
entry point to execute the protocol.
As previously introduced, our design separates program logic and
communications, corresponding to the \emph{callbacks} type
(\cref{subsection:impl-handler}) and \emph{connection} type
(\cref{subsection:impl-connection}).
The generated entry point function takes callbacks and a connection, and runs
the protocol, which we detail the internals in
\cref{subsection:impl-run-fsm}.

\subsubsection{Callbacks}
\label{subsection:impl-handler}

We generate function types for handling transitions in the CFSM, and collect
them into a record of \emph{callbacks}.
When running the CFSM for a participant, appropriate callback will be invoked
when the transition occurs.
For sending transitions, the sending callback is invoked to prompt a value to
be sent.
For receiving transitions, the receiving callback is invoked with the value
received, so the developer may use the value for processing.

\mypara{Generating Handlers for Receiving States}
For a receiving state $q \in \Q$ with receiving action $\ltsmsg{p}{q}{l}{x}{T}$
(assuming the current role is $\ppt q$), the receiving callback is a function
that takes the record $\enc{q}$ and the received value of type $T$, that
returns \tunit\ (with possible side-effects).
The function signature is given by
$$\code{state}q\code{\_receive\_}\dlbl{l}: (\dexp{st}:\enc{q})
  \rarr
\enc{\dte T}_{\dexp{st}}
  \rarr
ML~\tunit$$
The constructor $ML$ is an effect specification in \fstar, which permits all
side effects in ML language (e.g.\ using references, performing I/O), in
contrast to a pure total function permitting no side effects.
$\enc{q}$ is a record correspondent to the local typing context of the state,
$\enc{\dte T}_{\dexp{st}}$ is a refinement type, but the free variables in
the refinement types are bound to the record $\dexp{st}$.
We generate one callback for each receiving action, so that it can be invoked upon
receiving a message according to the message label.

\mypara{Generating Handlers for Sending States}
For a sending state $q \in \Q$ with send actions
$\ltsmsg{p}{q}{l_i}{x_i}{T_i}$ (assuming the current role is $\ppt p$), for
some index $i \in I$,
the sending callback is a function that takes the record $\enc{q}$, and
returns a disjoint union of allowed messages (with possible side-effects).
The constructor of of the disjoint union determines the label of the message,
and takes the payload as its parameter.
The function signature is given by
$$\code{state}q\code{\_send} : (\dexp{st}:\enc{q})
  \rarr
ML~\biguplus_{i \in I}\dlbl{l_i}(\enc{\dte T_i}_{\dexp{st}})$$
Different from receiving callbacks, only one sending callback  is generated for each
sending state.
This corresponds to the nature of internal choices, that the process implementing
a sending prefix makes a single selection;
on the contrary, the process implementing a receiving prefix must be able to
handle all branches.

\begin{remark}[Handlers and LTS Transitions]
  \upshape
  If the callback returns a choice with the refinements satisfied, the CFSM is
  able to make a transition to the next state.
  When a callback is provided, against its prescribed type, then the type of the
  callback type is inhabited and we can invoke the callback to obtain the label and
  the value of the payload.
  A callback function type may be uninhabited, for instance, when none of the choices
  are applicable. In this case, the endpoint cannot be implemented (we show an
  example below).
  If the developer provides a callback, then it must be the case that the
  specified type is inhabited.
  In this way, we ensure the protocol is able to make progress, and is not stuck
  due to empty value types\footnotemark.
\end{remark}

\footnotetext{Since we use a permissive $ML$ effect in the callback type, the
  callback may throw exceptions or diverge in case of unable to return. Such
  behaviours are not in the scope of our interest when we talk about progress.}

\subsubsection{Connections}
\label{subsection:impl-connection}

The \emph{connection} type is a record type with functions for sending and
receiving base types.
The primitives for communications are collected in a record with fields as
follows ($\dte S$ range over base types \tint, \tbool, \tunit, etc.):
\[
  \arraycolsep=1pt
  \begin{array}{lclllcl}
    \code{send\_}\dte{S} & : & \enc{\bigP} \rarr \enc{\dte{S}} \rarr ML~\tunit &
    \hspace{15mm}\ &
    \code{recv\_}\dte{S} & : & \enc{\bigP} \rarr \tunit \rarr ML~\enc{\dte{S}}
  \end{array}
\]
where $\enc{\bigP}$ is a disjoint union of participants roles and
$\enc{\dte{S}}$ is the data type for $\dte S$ in the programming language.
The communication primitives do not use refinement types in the type
signature.
We can safely do so by exploiting the property that refinements can be erased
at runtime after static type-checking.

\subsubsection{State Records with Refinements}
\label{subsection:impl-typedef-st}

We generate a type $\enc{q}$ for each state $q \in \Q$ in the CFSM.
The type $\enc{q}$ is a record type corresponding to the local typing context
in the state.
For each variable in the local typing context, we define a field in the
encoded record type, corresponding to the refinement type in the typing
context.
Since refinement types allow dependencies in the typing context, we exploit
the feature of dependent records in \fstar to encode the dependencies.

We use the smallest typing context associated with the CFSM state for the
generated record type.
The typing context can be computed via a graph traversal of the CFSM,
accumulating variables in the local type prefix along the traversal.

\subsection{Verified Endpoint Implementation (User Input)}
To implement an endpoint, a developer needs to provide a record of type
\emph{callback}, containing individual callbacks for transitions, and a record
of type \emph{connection}, containing functions to send and receive values of
different base types.
The two records are passed as arguments to the entry point function
\code{run} to execute the protocol.

The design of connection record gives freedom for the developer to implement
any transport layer satisfying first-in-first-out (FIFO) delivery without
loss of messages.
These assumptions originate from the network assumptions in MPST.
TCP connections are a good candidate to connect the participants in the
protocol, since they satisfy the assumptions.

To satisfy the data dependencies as specified in the protocol, the provided
callbacks  must match the generated refinement type.
The \fstar compiler utilises the Z3 SMT solver \cite{TACAS08Z3} for type-checking,
saving developers the need for manual proofs of arithmetic properties.
After type-checking, the compiler can extract the \fstar code into OCaml (or
other targets), where the refinements at the type level are erased.
Developers can then compile the extracted OCaml program to execute the protocol.

The resulting program has data dependencies verified by \fstar using
refinement types.
Moreover, the MPST theory guarantees that the endpoints are free for
deadlocks or communication mismatches, and conform to the global types.

\label{subsection:impl-endpoint}